\documentclass[
reprint,amsfonts, amssymb, amsmath,  showkeys,pra, superscriptaddress, twocolumn,longbibliography,nofootinbib,notitlepage]{revtex4-2}

\usepackage{tikz}
\usetikzlibrary{fadings,arrows.meta}

\usepackage{graphicx}
\usepackage{array}
\usepackage{xcolor}
\usepackage{float}
\usepackage{physics}
\usepackage{enumerate}
\usepackage{enumitem}
\usepackage{lipsum}
\usepackage{amsfonts}
\usepackage{mathtools}
\usepackage{dsfont}
\usepackage{bbm}
\usepackage[normalem]{ulem}

\usepackage{booktabs}
\usepackage{url}
\usepackage{braket}

\usepackage{amsthm}
\usepackage{bm}
\usepackage[T1]{fontenc} 

\definecolor{C4}{RGB}{202, 21, 81}
\usepackage[colorlinks=true,citecolor=C4,linkcolor=C4,urlcolor=C4]{hyperref}
 
\usepackage{thm-restate}

\usepackage{tikz}
\usetikzlibrary{fadings}
\usepackage{mathrsfs}

\newcommand{\sdket}[1]{| #1 \rangle\!\rangle}

\definecolor{flowNavy}{HTML}{18324A}
\definecolor{flowBlue}{HTML}{061A40}
\definecolor{flowTeal}{HTML}{247BA0}  
\definecolor{flowGreen}{HTML}{70C1B3}  
\definecolor{flowLime}{HTML}{9AD5B3} 
\definecolor{flowRed}{HTML}{BA3B46}
\definecolor{flowInk}{HTML}{23313C}
\definecolor{flowMuted}{HTML}{5F6C75}
\definecolor{flowBluePale}{HTML}{ECF2FD} 
\definecolor{flowTealPale}{HTML}{EEF7FB} 
\definecolor{flowGreenPale}{HTML}{F1F9F8} 
\definecolor{flowLimePale}{HTML}{F1F9F4} 
\definecolor{flowRedPale}{HTML}{FAEFF0}

\providecommand{\learningflowbullet}[1]{%
  \par\noindent
  \hangindent=.2em
  \hangafter=0
  \makebox[1.2em][l]{$\bullet$}#1\par
}

\usepackage{comment}

\theoremstyle{remark}

\newcommand*{\ot}{\otimes}

\DeclareMathOperator*{\expect}{\mathbb{E}}

\newcommand{\x}{\boldsymbol{x}}
\newcommand{\s}{\boldsymbol{s}}

\newcommand{\mst}{\mathsf{T}}

\definecolor{antonio}{rgb}{.2,.5,.1}

\renewcommand{\arraystretch}{1.5}

\newcommand{\mcf}{\mathcal{F}}
\newcommand{\mcg}{\mathcal{G}}

\newcommand{\mcc}{\mathcal{C}}

\newcommand{\mcx}{\mathcal{X}}

\newcommand{\mco}{\mathcal{O}}

\newcommand{\mce}{\mathcal{E}}

\newcommand{\mbc}{\mathbb{C}}

\newcommand{\mbz}{\mathbb{Z}}
\newcommand{\mbf}{\mathbb{F}}
\newcommand{\mbe}{\mathbb{E}}

\newcommand{\mbu}{\mathbb{U}}

\newcommand{\id}{\mathds{1}}

\def\tt{^{\otimes t}}

\newcommand{\stab}{\text{Stab}}

\def\a{\alpha}
\def\b{\beta}

\def\d{\delta}
\def\sg{\sigma}

\def\ad{^{\dagger}}

\usepackage{float}
\floatstyle{ruled}
\newfloat{algorithm}{tbp}{loa}
\floatname{algorithm}{Algorithm}

\newcounter{algstep}

\newcounter{algsubstep}

\begin{document}
\title{Sample-optimal learning of stabilizer states} 

\author{Rebecca Chang}
\affiliation{Theoretical Division, Los Alamos National Laboratory, Los Alamos, New Mexico 87545, USA}
\affiliation{Massachusetts Institute of Technology, Cambridge, Massachusetts 02139, USA}
\affiliation{Department of Computer Science, University of Warwick, Coventry, UK}

\author{Matthias C. Caro}
\affiliation{Department of Computer Science, University of Warwick, Coventry, UK}

\author{Mart\'{i}n Larocca}
\affiliation{Theoretical Division, Los Alamos National Laboratory, Los Alamos, New Mexico 87545, USA}
\affiliation{Quantum Science Center, Oak Ridge, TN 37931, USA}

\author{Maxwell West}
\affiliation{Theoretical Division, Los Alamos National Laboratory, Los Alamos, New Mexico 87545, USA}
\affiliation{Quantum Science Center, Oak Ridge, TN 37931, USA}

\begin{abstract}
It is well-known that learning a pure $n$-qubit   stabilizer state $\ket\psi$ both requires, and can be accomplished with, access to a number of copies of $\ket\psi$ linear in $n$. However, the precise constant coefficient of this  scaling does not appear to have been determined. Here we prove that  $L_\delta(n)$, the smallest number of copies from which a quantum procedure can identify any stabilizer state  with failure probability at most $0<\delta<1/8$, satisfies   
$n+\lceil\log_2(1/\delta)\rceil-3\leq L_\delta(n)\leq n+\left\lceil\log_2(1/\delta)\right\rceil+4$. We present a polynomial-time quantum learning algorithm that saturates this bound, achieving a constant factor improvement in sample-complexity over previously known approaches. As an immediate corollary, we obtain via the Choi-Jamiolkowski isomorphism an algorithm for learning an unknown $n$-qubit Clifford unitary from $2n+\left\lceil\log_2(1/\delta)\right\rceil+4$ queries,  the $n$-dependence of which we show to be optimal. Our proof technique, which  involves Fourier analysis on the abelian group $\mbz_4^n \times \mathbb{F}_2^{n(n-1)/2}$, seems to be qualitatively different to previous approaches to stabilizer state learning, and may   be of some independent interest; in particular, it admits  natural generalisations to further problems in quantum learning theory. 
\end{abstract}

\maketitle


\tableofcontents

\section{Introduction}\label{sec:intro}
The Clifford group is ubiquitous in quantum information theory because it constitutes a highly non-trivial family of classically simulable dynamics with remarkable statistical properties~\cite{aaronson2004improved,zhu2016clifford}. Its orbit on a computational basis state, the set of so-called \textit{stabilizer states}, is likewise pervasive and for example important in the theory of  quantum error correction \cite{gottesman1997stabilizer, roffe2019quantum}. Their algebraic features  also make them a popular test case for quantum learning theory, leading to several natural questions. For example, the property testing of ``stabilizerness'' (that is, determining whether an unknown  state is a stabilizer state) has been shown to be achievable given access to a constant  number of copies of the state~\cite{gross2021schur,bittel2026complete}. \\

In particular, the \textit{learnability} of stabilizer states has  been investigated with some enthusiasm (see Table~\ref{tab:prior_learners})~\cite{gottesman2008identifying,montanaro2017learning,allcock2025reconquering,hinsche2026abelian,arunachalam2026optimal,arunachalam2022optimal}; as opposed to the intractable task of learning a general quantum state \cite{scharnhorst2025optimal,bruss1999optimalstate, haah2017sample, o2016efficient}, an $n$-qubit stabilizer state can be learnt with a number of samples merely linear in $n$~\cite{gottesman2008identifying,montanaro2017learning}. This linear scaling is asymptotically  optimal   for information-theoretic reasons: there are $2^{n^2/2+\mco(n)}$ stabilizer states~\cite{aaronson2004improved}, so specifying one requires $\sim n^2/2$ bits of information. As an $n$-qubit quantum state cannot convey more than $n$ bits of information by Holevo's bound \cite{holevo1973bounds}, one needs to process at least $n/2$ copies.
An algorithm for actually doing so was first developed  by Aaronson and Gottesman, who in fact presented both an algorithm which used $\mco(n^2)$ samples but only single-copy measurements, and an $\mco(n)$-copy procedure with a collective measurement~\cite{gottesman2008identifying}. The situation was later improved upon by Montanaro, who  gave an $\mco(n)$-copy algorithm via Bell sampling, with measurements on at most two copies at a time~\cite{montanaro2017learning}. A more recent $\mco(n^2)$-copy algorithm, by Arunachalam et al~\cite{arunachalam2022optimal}, managed to employ only separable measurements. 
Other recent works have extended things in several directions, such as learning more broadly    qudit stabilizer states~\cite{allcock2025reconquering}, connecting the problem to the abelian state hidden subgroup problem~\cite{hinsche2026abelian}, restricting the allowed quantum memory~\cite{arunachalam2026optimal}, and learning states prepared by Clifford+$T$ circuits with low $T$ count~\cite{grewal2025efficient, hangleiter2024bell, leone2024learningt}. \\

Despite these various developments, however, it appears that  the optimal constant prefactor in the  sample-complexity scaling of stabilizer state learning has remained undetermined. In this work, we show that this optimal leading coefficient is exactly one; indeed, for a fixed failure probability $0<\delta<1/8$, the minimum number of copies required is, up to a small additive constant, exactly $n+ \log1/\delta $. Our proof of the corresponding upper bound involves the construction of   an explicit polynomial-time algorithm, which achieves a constant-factor improvement over, for example, the Bell sampling algorithm given by Montanaro~\cite{montanaro2017learning}, which requires $3n+\mco(\log(1/\delta))+2$ samples (and which can be easily reduced to $2n+\mco(\log(1/\delta))+3$ if one is willing to make a slightly more general measurement, see Appendix~\ref{sec:bell}). \\

\begin{figure*}[t]
    \centering

    \begin{tikzpicture}[
  x=1cm,
  y=1cm,
  font=\sffamily\footnotesize,
  text=flowInk,
  line join=round,
  line cap=round,
  box/.style={
    rectangle,
    rounded corners=2mm,
    draw=flowBlue,
    fill=flowBluePale,
    line width=0.8pt,
    minimum height=18mm,
    align=center,
    inner sep=2.5mm
  },
  compress/.style={
    box,
    draw=flowTeal,
    fill=flowTealPale
  },
  fourier/.style={
    box,
    draw=flowGreen,
    fill=flowGreenPale
  },
  output/.style={
    box,
    draw=flowLime,
    fill=flowLimePale,
    line width=1pt
  },
  retry box/.style={
    rectangle,
    rounded corners=1.5mm,
    draw=flowRed,
    fill=flowRedPale,
    line width=0.7pt,
    minimum height=9mm,
    text width=37mm,
    align=center,
    inner sep=1.5mm,
    font=\sffamily\scriptsize
  },
  main arrow/.style={
    -{Latex[length=2.2mm,width=1.5mm]},
    draw=flowNavy,
    line width=0.75pt
  },
  retry arrow/.style={
    -{Latex[length=2mm,width=1.4mm]},
    draw=flowBlue,
    line width=0.7pt
  },
  edge label/.style={
    inner sep=0pt,
    font=\sffamily\scriptsize\bfseries
  }
]

\node[anchor=west, font=\sffamily\small, text=flowInk] at (0.15,7.15) { 
};

\node[box, text width=42mm, minimum height=26mm, align=left] (search) at (2.55,3.45) {
  \makebox[\linewidth][c]{\textbf{1.\ Clifford-chart search}}\par\vspace{1.5mm}
  \learningflowbullet{Apply a random Clifford $C^{\ot t}$ }
  \vspace{1mm}
  \learningflowbullet{Reversibly test whether basis\\  \ \hspace{2.375mm} strings affinely span
  $\mathbb F_2^n$}
};

\node[compress, text width=39mm, align=left] (compression) at (8.25,4.25) {
  \makebox[\linewidth][c]{\textbf{2.\ Isotypic compression}}\par\vspace{0.5mm}
  \learningflowbullet{Compute  character $h$}
  \vspace{1mm}
  \learningflowbullet{Compress along isotypics:}
  \vspace{1mm}
  \noindent\makebox[\linewidth][c]{\scriptsize
    $\lvert\mathcal F_h\rangle\lvert h\rangle
    \mapsto\lvert0^{nt}\rangle\lvert h\rangle$
  }\par
};

\node[fourier, text width=34mm, align=left] (qft) at (13.05,4.25) {
  \makebox[\linewidth][c]{\textbf{3.\  Fourier transform}}\par\vspace{0.5mm}
  \learningflowbullet{Want label   $(\alpha,\beta)\in \Upsilon$}
  \vspace{1mm}
  \learningflowbullet{$\Upsilon$-QFT\textsuperscript{-1}\hspace{-0.45mm} recovers\hspace{-0.5mm} $(\alpha,\beta)$}
};

\node[output, text width=25mm] (tableau) at (16.90,4.25) {
  \textbf{4.\ Output}\\[0.5mm]
  \vspace{2mm}
  Calculate   tableaux of
  $C^\dagger\lvert\phi_{\alpha,\beta}\rangle$
};

\node[retry box, text width=40mm] (retry) at (8.25,2.5) {
  Uncompute rank and $C$, go to  next trial
  (up to $r$ trials)
};

\draw[main arrow, dashed]
  ([yshift=8.mm]search.east)
  -- node[edge label, above=1mm, text=flowInk] {accept}
  (compression.west);
\draw[main arrow, dashed]
  ([yshift=-9.4mm]search.east)
  -- node[edge label, above=1mm, text=flowInk] {reject}
  (retry.west);
\draw[main arrow] (compression.east) -- (qft.west);
\draw[main arrow] (qft.east) -- (tableau.west);

\draw[retry arrow]
  (retry.south) -- (8.25,1.25) -- (2.55,1.25) -- (search.south);

\pgfresetboundingbox
\path[use as bounding box] (0,0.25) rectangle (18.60,7.50);

\end{tikzpicture}
\caption{Overview of our stabilizer-learning algorithm, which has four core stages. In stage one we reduce the problem to that of learning a stabilizer state which has support on all computational basis strings; such ``full-rank'' states enjoy certain nice algebraic properties which we are able to exploit. In particular, they are exactly the orbit on the state $\ket{+}^{\ot n}$ of a certain representation of the abelian group $ \Upsilon \cong \mbz_4^n \times \mathbb{F}_{\!2}^{\,\smash[b]{n(n-1)/2}}  $. The test (correctly) rejecting a rank-deficient state is non-destructive, so that the test can be uncomputed and another trial performed; the test (incorrectly) rejecting a full-support causes the algorithm to fail (not depicted in the figure).
In stage two we ``compress'' our state to one which lives in the group algebra of $\Upsilon$; doing this efficiently is one of the key technical   challenges. In stage three we perform an inverse QFT over $\Upsilon$, which we show reveals (with high probability) exactly the information needed to recover our initial unknown state. Finally, a little classical postprocessing suffices to reconstruct the stabilizer tableaux of the original state from the data $(C,(\a,\b))$.}
    \label{fig:1}
\end{figure*}
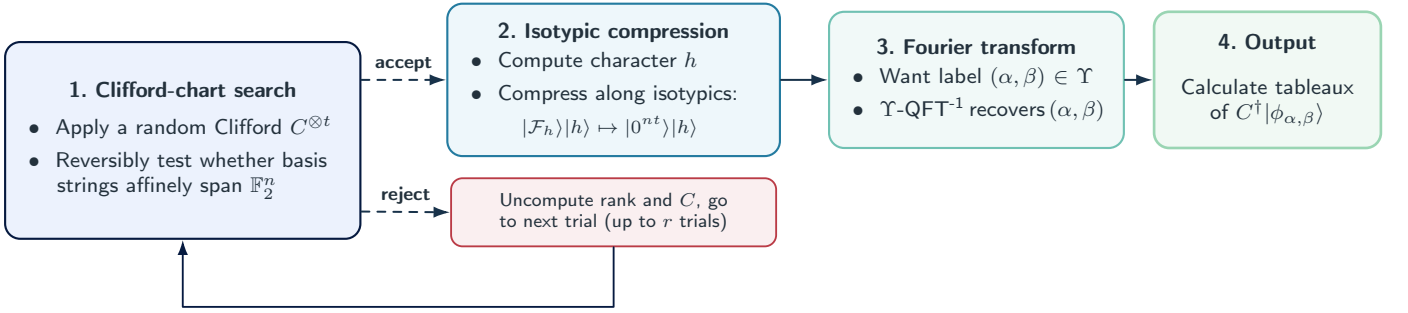

Somewhat more formally, for $0\leq \delta \leq 1$, let $L_\delta(n)$ be the smallest number of copies from which some quantum procedure (which can include adaptivity, ancillae, collective measurements etc.) can identify any state in $\stab_n$, the set of $n$-qubit stabilizer states, with worst-case failure probability at most $\delta$. 
Our main result is:
\begin{restatable}{thm}{main}\label{thm:main}
For every $0<\delta<1/8$,
\begin{equation}
    n+\lceil\log_2(1/\delta)\rceil-3\leq L_\delta(n)\leq n+\left\lceil\log_2(1/\delta)\right\rceil+4 .
\end{equation}
Furthermore, the upper bound is achievable by a quantum algorithm which runs in time $\mco(n(n+\log(1/\delta))^4)$.
\end{restatable}

\noindent
As one can learn an $n$-qubit Clifford unitary by learning its Choi state, which is itself manifestly a $2n$-qubit stabilizer state,  we obtain the immediate corollary
\begin{restatable}{crl}{crlcliff}~\label{crl:cliff}
For every $0<\delta<1/8$, one can learn an $n$-qubit Clifford via $2n+\left\lceil\log_2(1/\delta)\right\rceil+4$ (forward) queries. 
\end{restatable}
We show in Appendix~\ref{sec:cliff} that the $n$-dependence of Corollary~\ref{crl:cliff} is optimal.

\section{Upper bound}\label{sec:upperbound}
In this section we describe our learning algorithm (see Figure~\ref{fig:1}), the analysis of which will provide a constructive proof of the upper bound of Theorem~\ref{thm:main}. To improve readability, we will at various points cite lemmas whose proofs are delegated to the appendices.

\subsection{Learning algorithm overview}
We now walk through our  stabilizer learning algorithm; as depicted in Figure~\ref{fig:1}, there are four  main steps: (i) reducing the   problem to that of learning a stabilizer state of full support across the computational basis, (ii) compressing the resulting state to one living in the group algebra $\mbc[\Upsilon]$ of a particular abelian group $\Upsilon$,  (iii) performing a Fourier transform over $\Upsilon$ to extract some relevant information, and (iv) converting the learnt information into an explicit stabilizer tableau representing the unknown state. 

\begin{table}[t]
\label{tab:prior_learners}
\centering
\footnotesize
\setlength{\tabcolsep}{3pt}
\renewcommand{\arraystretch}{1.18}
\begin{tabular*}{\columnwidth}{@{\extracolsep{\fill}}lccc@{}}
\toprule
Approach & Copies & Coherence & Time \\
\midrule
Aaronson--Gottesman~\cite{gottesman2008identifying}
  & $\mco(n^2)$ & $1$ & $\operatorname{poly}(n)$ \\
Aaronson--Gottesman~\cite{gottesman2008identifying}
  & $\mco(n)$ & $\mco(n)$ & $\operatorname{poly}(n)$ \\
Bell sampling~\cite{montanaro2017learning} 
  & $3n+\mco(1)$ & $2$ & $\mco(n^3)$ \\
\quad + joint sign readout
  & $2n+\mco(1)$ & $2$ & $\mco(n^3)$ \\
Separable measurements~\cite{arunachalam2022optimal}
  & $\mco(n^2)$ & $1$ & $2^{\mco(n^2)}$ \\
StateHSP~\cite{hinsche2026abelian}
  & $\mco(n)$ & $2$ & $\operatorname{poly}(n)$ \\
\midrule
This work
  & $n+\mco(1)$ & $n+\mco(1)$ & $\widetilde \mco(n^5)$ \\
\bottomrule
\end{tabular*}
\caption{The sample complexities of previous approaches to learning  learning  an unknown pure $n$-qubit stabilizer state   at fixed failure probability $0<\delta<1/8$, along with how many copies need to be acted on coherently, and the algorithm time.}
\end{table}

\subsubsection{Reduction to a full-support state}
In step one, we reduce the problem to that of learning stabilizer states which are in a certain sense ``full rank''. Our starting point is  that each stabilizer state may be written in the   form~\cite{dehaene2003the,gross2008the}
\begin{equation}
    \ket{S}=2^{-k/2}\sum_{u\in L}i^{f(u)}\ket{a+u}\,,
\end{equation}
where $L$ is a dimension-$k$ subspace of  $\mathbb{F}_2^n$,  $a\in\mbf_2^n$, and  $f:L\to \mbz_4$ is a quadratic function. As we shall see, this form has a particularly nice interpretation when $k=n$; in this case, we say that $\ket S$ is of \textit{full rank}. \\

Now, for $t\geq n+1$, let $\Pi_{n,t}$ be the orthogonal projector onto the subspace spanned by those computational basis tuples $\x=(\x^{(1)},\dots,\x^{(t)})\in (\mathbb{F}_2^n)^t$ whose affine span is all of $\mathbb{F}_2^n$, i.e. $\text{affrank}(\x)=n$ where
\begin{equation}
\text{affrank}(\x):=\text{rank}[\x^{(2)}-\x^{(1)},\dots,\x^{(t)}-\x^{(1)}]\,.\label{eq:affrank}
\end{equation}
This projector can be implemented reversibly and efficiently by coherently computing (via Gaussian elimination, say) the rank  into an ancilla, measuring that ancilla, and uncomputing the rank. To be concrete, denote by $b(\x)=\mathds{1}_{\text{affrank}(\x)=n}$ an indicator bit for whether the tuple $\x$ is of full rank in the sense of Eq.~\eqref{eq:affrank}, and write $\ket{S}\tt = \sum_{\x\in (\mathbb{F}_2^n)^t}\a_{\x} \ket{\x} $ for some  $\a_{\x}\in\mbc$. The progression to this point then looks like
\begin{align*}
    \sum_{\x\in (\mathbb{F}_2^n)^t}\a_{\x} \ket{\x}\ket 0&\to \sum_{\x\in (\mathbb{F}_2^n)^t}\a_{\x} \ket{\x}\ket {b(\x)}\\
    &\hspace{-3.25mm}\underset{{\rm measure}}{\to} \sum_{\substack{\x\in (\mathbb{F}_2^n)^t\\ b(\x)=b}}\a_{\x} \ket{\x}\ket {b}\\
    &\to \sum_{\substack{\x\in (\mathbb{F}_2^n)^t\\b(\x)=b}}\a_{\x} \ket{\x}\ket {0} 
\end{align*}

Crucially, if $k<n$ then $\Pi_{n,t}\ket{S}^{\otimes t}=0$ (i.e. $b(\x)=0$ for all $\x$ with  $\a_{\x}\neq 0$), so that measuring the ancilla gives rejection with certainty, and neither the test nor the uncomputation step disturb the state\footnote{It might seem   simpler to instead employ the POVM which tests if $\ket{S}^{\ot t}$ lies within the span of the $t$-fold tensor powers of the stabilizer states of rank less than $n$;  it is however unclear if that POVM can be implemented efficiently.}. Conversely, for a full-support state we find:

\begin{restatable}{lem}{fullsupport}\label{lem:fullsupp}
The $t$-fold tensor power of a full-support stabilizer state is accepted by $\Pi_{n,t}$ with  probability  
\begin{equation}
a_{n,t}=\prod_{j=0}^{n-1}(1-2^{j-(t-1)})  \,,
\end{equation}
  which satisfies $1-a_{n,t}<2^{n-t+1}$.
\end{restatable}

In addition, a random Clifford\footnote{It would in fact prove sufficient to sample randomly from the right coset $\Upsilon\backslash\mathsf{Cl}_n$; this is a minor distinction we do not belabour.} takes us to a full-support state with constant probability:
\begin{restatable}{lem}{fullprob}\label{lem:fullprob}
Let $\ket{S}$ be an arbitrary stabilizer state. If $C$ is a uniformly random Clifford, then $C\ket{S}$ has full computational-basis support with probability $p_n=\prod_{j=1}^n (1+2^{-j})^{-1}>\frac{1}{3}$.
\end{restatable}
Taken together, given access to a number $t$ of copies slightly greater than $n$  of $\ket{S}$, Lemmas~\ref{lem:fullsupp} and~\ref{lem:fullprob} allow us to (with high probability) obtain a  full-rank state related to our original state by a known Clifford $C$. Indeed, we merely sample a random Clifford, apply its $t$-fold tensor power, and check if we obtained a full rank state; if not, we apply the $t$-fold tensor power of its inverse and try again with a new random Clifford. The  probability of failure, which occurs when either we  ``incorrectly'' reject a full rank state, or none of the $r$ sampled Cliffords yields a full-support state, is discussed in Section~\ref{sec:accounting}.  \\

Now, having arrived at  a full-rank state, we   have $a+L=L=\mathbb{F}_2^n$, so that we can write such a state as
\begin{equation}
\ket{\phi_{\alpha,\beta}} = 2^{-n/2}\sum_{x\in{\mathbb{F}_2^n}}i^{\alpha\cdot x + 2\sum_{1\leq j<k\leq n}\beta_{jk}x_j x_k}\ket{x}\,,
\end{equation}
labeled by $\alpha\in\mbz_4^n$ and $\beta\in\mathbb{F}_2^{\binom{n}{2}}$. Identifying the state is then clearly equivalent to recovering $\alpha$ and $\beta$; note that this would not be true without the full-rank assumption (while one could still introduce parameters $\a$ and $\b$ to describe the quadratic form in the rank-deficient case, learning them would for example not reveal the support of the state). The full-rank states, though, are thus exactly specified by elements of the abelian group\footnote{Incidentally, a group structure is indeed present (that is, this is not merely an identification on the level of sets), indeed, note $D_{\alpha,\beta}D_{\alpha',\beta'}=D_{\alpha+\alpha',\,\beta+\beta'}$;  this will be important later.} $\Upsilon := \mbz_4^n \times \mathbb{F}_2^{\binom{n}{2}}$ of size $|\Upsilon| = 2^{(n^2+3n)/2}$. Explicitly, with 
\begin{equation}\label{eq:dab}
D_{\alpha,\beta}:=\prod_{j=1}^n S_j^{\alpha_j}
\prod_{1\leq j<k\leq n}
\operatorname{CZ}_{jk}^{\beta_{jk}}
\end{equation}
for $(\a,\b)\in\Upsilon$ we have $\ket{\phi_{\alpha,\beta}}=D_{\alpha,\beta}\ket{+}^{\otimes n}$. This identification completes the first step of our algorithm.

\subsubsection{Isotypic compression}
After the first stage  of the algorithm, our state is of the form $\ket{\psi_{{\rm stage}\,2}}=a_{n,t}^{-1/2}\Pi_{n,t}\ket{\phi_{\alpha,\beta}}^{\otimes t}$ for some unknown $(\a,\b)\in\Upsilon$, which we would like to determine. 
Write a computational basis tuple as an $n\times t$ binary matrix $\x$, and let $\x_j\in \mathbb{F}_2^t$ be its $j$th row. We   define~\cite{wood1993witts,gross2021schur}
\begin{equation}\label{eq:h}
    h(\x)=((q_t(\x_j))^n_{j=1},(\x_j\cdot \x_k)_{1\leq j<k\leq n})\,,
\end{equation}
where $ q_t(x)=|x|\mod 4$ is the Hamming weight of the bitstring $x\in\mbf_2^t$ (mod 4). The relevance of taking the residue modulo four is that the phase  of the basis vector $\x$ in $\ket{\phi_{\alpha, \beta}}^{\otimes t}$ may be written as  the $\Upsilon$-character
\begin{equation}
\chi_{h(\x)}((\alpha,\beta))=i^{\sum_j \alpha_j q_t(\x_j)+2\sum_{j<k}\beta_{jk}(\x_j\cdot \x_k)}\,.
\end{equation}
Now (recalling Eq.~\eqref{eq:affrank}) let us define the set
\begin{equation}
    \mcf_h = \{\x:h(\x)=h\text{ and }\text{affrank}(\x)=n\}\,;
\end{equation}
note that the span of the $\ket{\x}$ corresponding to the constituents $\x$ of a given $\mcf_h$ is exactly an isotypic component of the $\Upsilon$-representation furnished by the full affine rank strings. Now, the phase of each basis tuple depends on $\x$ only through $h(\x)$, so if two tuples $\x, \x'$ are in the same isotypic (i.e. $h(\x)=h(\x')$), they have identical amplitude. Therefore, we can rewrite our state by grouping computational basis states by isotypic:
\begin{align*}
\hspace{-2mm}\frac{\Pi_{n,t}\ket{\phi_{\alpha,\beta}}^{\otimes t}}{\sqrt{a_{n,t}}}
&=\frac{1}{\sqrt{a_{n,t}2^{nt}}}
\sum_{\substack{\text{affrank}(\x)=n}}
\hspace{-2mm}\chi_{h(\x)}((\alpha,\beta))\ket \x\\
&=\sum_h\chi_{h}((\alpha,\beta))\sqrt{\frac{|\mathcal F_h|}{a_{n,t}2^{nt}}}\,\ket{\mathcal F_h}\\
&=\sum_h\sqrt{p_h^{\mathrm{aff}}}\,\chi_{h}((\alpha,\beta))\ket{\mathcal F_h},
\end{align*}
with $p_h^{\mathrm{aff}}=|\mathcal F_h|/(a_{n,t}2^{nt})$ and $\ket{\mathcal F_h}=\frac{1}{\sqrt{|\mathcal F_h|}}
\sum_{\x\in\mathcal F_h}\ket \x$. 
At this point there are two relevant ``degrees of freedom'' in our state. We have a sum over various values of $h$, and, for a given $h$, a sum over the corresponding isotypic. This latter freedom  concerns us not, and indeed we aim now to efficiently ``compress'' each such isotypic into a single  state. This desire is formalised and achieved by the following lemma:

\begin{restatable}{lem}{unitary}\label{lem:unitary}
For any $n$, $t\geq n+1$ and $\zeta>0$, we have a circuit implementing $\sum_{h} \ketbra{h} \ot U_h$ of depth
$\mco(nt(t+\log(nt/\zeta))^3)$ such that $\left\|U_h\ket{0^{nt}}-\ket{\mathcal F_h}\right\|_2\leq\zeta\,$ for all $h$ corresponding to non-empty fibers. 
\end{restatable}
Having computed $h(\x)$ into a clean register, then, Lemma~\ref{lem:unitary} allows us to coherently apply $U_h^\dagger$ to the $\x$ register,  yielding in superposition the transformation $\ket{\mathcal{F}_h}\ket{h}\mapsto \ket{0^{nt}}\ket{h}$. 
This concludes the isotypic compression step.

\subsubsection{Fourier transform}
At the beginning of  the final part of our algorithm, our state has assumed the form
\begin{equation}\label{eq:qftstart}
    \ket{\psi_{{\rm stage}\,3}}=\sum_{h\in \widehat \Upsilon}\sqrt{p^{\text{aff}}_h}\chi_{h}((\alpha,\beta))\ket{h}\,.
\end{equation}
Now, with
\begin{equation}
    \mathsf{F}_{\Upsilon}=\frac{1}{\sqrt{|\Upsilon|}}\sum_{(a,b)\in\Upsilon,h\in\widehat{\Upsilon}}\chi_h ((a,b))\ketbra{h}{a,b}
\end{equation}
the Fourier transform over $\Upsilon$, applying $\mathsf{F}_{\Upsilon}^\dagger$ to Eq.~\eqref{eq:qftstart} yields a label $(a,b)\in{\Upsilon}$ with probability 
\begin{align}
    p_{a,b}= \frac{1}{|\Upsilon|}\Big\lvert\sum_h \sqrt{p_h^{\text{aff}}} \chi_{h}((\alpha,\beta))\overline{\chi_h((a,b))}\Big\rvert^2.
\end{align}
If $(a,b)=(\alpha,\beta)$, the character phases cancel, and so measuring gives the (correct) label $(\alpha,\beta)$ with probability
\begin{equation}\label{eq:psum}
    p_{\rm succ}=\frac{1}{|\Upsilon|}\Big(\sum_h\sqrt{p^{\text{aff}}_h}\Big)^2.
\end{equation}
Now, if $p_h^{\text{aff}}$ were uniform, the state in Eq. \eqref{eq:qftstart} would be exactly $\mathsf{F}_{\Upsilon}\ket{\alpha,\beta}$, so that decoding would succeed with probability $1$;  in reality  there is some nonuniformity present  in $p_h^{\text{aff}}$ for finite $t$, leading to some nonzero failure probability. It thus remains to control the sum of Eq.~\eqref{eq:psum}; we find (Appendix~\ref{sec:minutiae}) that it is very well behaved for $t$ marginally greater than $n$, leading to:
\begin{restatable}{lem}{fourier}\label{lem:fourier}
Conditioned on the affine rank test having accepted some full-support state $\ket{\phi_{\alpha, \beta}}^{\otimes t}$, for $t=n+s$ with $s>1$, Fourier decoding returns $(\alpha, \beta)$ with probability 
\begin{equation}\label{eq:succbound}
    p_{\rm succ}\geq \frac{(1-2^{1-s})^2}{1+4/(2^{s-1}-1)}.
\end{equation}
\end{restatable}

As the right hand side of Eq.~\eqref{eq:succbound} is $1-\mco(2^{-s})$, we see that taking moderately more than $n$ copies of the state suffices to learn it with high probability; this intuition is formalised in    Section~\ref{sec:accounting}. 

\subsubsection{Tableau reconstruction}

Finally, let us explicitly detail how one can efficiently compute the \textit{stabilizer tableau}~\cite{aaronson2004improved,dehaene2003the} of our unknown state $\ket\psi$ from the data  $(C, (\alpha,\beta))$ that we have obtained throughout the course of the algorithm. That is, we wish to  find  $n$  algebraically independent (signed) Pauli strings for which $\ket\psi$ is an eigenvector with eigenvalue 1. To begin, we note that, with $D_{\alpha,\beta}$ as in Eq.~\eqref{eq:dab},
we have $\ket{\phi_{\alpha,\beta}} = D_{\alpha,\beta}\ket{+}^{\otimes n}$; as
$\ket{+}^{\otimes n}$ is stabilized by $X_1,\dots,X_n$, the state $\ket{\phi_{\alpha,\beta}}$ is then stabilized by the (manifestly algebraically independent) Pauli strings $g_j = D_{\alpha,\beta} X_j D_{\alpha,\beta}^\dagger$  for $j=1,\dots,n$.   As conditioned on the success of the algorithm, we furthermore have $\ket\psi = C\ad \ket{\phi_{\alpha,\beta}}$, it follows that $\ket\psi$ is itself stabilized by the operators 
\begin{equation}
    C^\dagger g_j C = (D_{\alpha,\beta}\ad C)\ad X_j (D_{\alpha,\beta}\ad C) \,.
\end{equation}
As $D_{\alpha,\beta}\ad C$ is a (known) Clifford operator, one can evaluate its adjoint action on the full set of the $X_j$ in time $\mco(n^3)$~\cite{aaronson2004improved}, yielding a generating set for the stabilizer of $\ket\psi$, and thereby identifying it uniquely.

\subsection{Resource accounting} \label{sec:accounting}
Putting everything together, we can obtain the following  bounds on the sample and computational-complexity of our algorithm:
\begin{restatable}{lem}{bounds}\label{lem:bounds}
With worst-case failure probability $0<\delta\leq 1$, our   algorithm   learns an unknown pure $n$-qubit stabilizer state  using $t=n+4+\bigl\lceil\log_2(1/\delta)\bigr\rceil$ input copies and at most $r=\lceil 3\ln(40/\delta)\rceil$   trials when searching for a Clifford that gives a full-support coordinate system.  
\end{restatable}
\noindent
Essentially, the $r$- and $t$- dependence come respectively from Lemma~\ref{lem:fullprob} and Lemma~\ref{lem:fourier}. On the computational complexity side, we find:
\begin{restatable}{lem}{compbound}\label{lem:compbound}
The computational complexity of our algorithm is $\mco(n(n+\log(1/\delta))^4)$.  
\end{restatable}
We find in the proof of Lemma~\ref{lem:compbound} that   the complexity of the algorithm is dominated by the orbit compression step, and is thereby determined by the result of Lemma~\ref{lem:unitary}. It turns out to be sufficient to take the approximation error $\zeta$ of the statement of that lemma to be of order $\delta$;   using also that $t\in\mco(n+\log(1/\delta))$, Lemma~\ref{lem:unitary} then yields Lemma~\ref{lem:compbound}.

\section{Lower bound}\label{sec:lb}
In order to establish the lower bound in Theorem~\ref{thm:main}, it suffices to  restrict the learning problem to the subset formed by full-support quadratic-phase states which appeared in the proof of the upper bound, namely
\begin{equation}
    \mathcal{E}_n:=\{\ket{\phi_{\alpha,\beta}}:(\alpha,\beta)\in \Upsilon\}\subseteq \text{Stab}_n\,.
\end{equation}
In fact, we will demonstrate that the claimed lower bound holds even for the easier task of learning (with average case failure probability at most $\delta$) states from $\mathcal{E}_n$, which immediately implies the same lower bound for the harder task of   learning arbitrary stabilizer states with worst case   failure  probability $\d$. \\

To begin, and recalling the    label $h(\x)$ (see Eq.~\eqref{eq:h}), we let $p^{(t)}$ be the ``unconditioned'' distribution of $h(\x)$, when $\x\in \mathbb{F}_2^{n\times t}$ is uniform, namely
\begin{equation*}
    p^{(t)}_h:=\Pr_{\x}[h(\x)=h]=2^{-nt}|\{\x \in \mathbb{F}_2^{n\times t}:h(\x)=h\}|
\end{equation*}
That is,  these probabilities are analogous to $p^{\text{aff}}_h$ from Eq.~\eqref{eq:qftstart}, but without the conditioning on having passed the affine-rank test. Now, the  first key observation is that  the states $\ket{\phi_{\a,\b}}\tt=2^{-nt/2}\sum_{\x}\chi_{h(\x)}((\a,\b))\ket{\x}$   are exactly the ``geometrically uniform'' orbit of a certain representation of $\Upsilon$; indeed, the one which acts as  $U_{\a,\b}\ket{\x}=\chi_{h(\x)}((\a,\b))\ket{\x}$. The utility of this observation is that it allows us to import ideas from the theory of \textit{pretty good measurements}~\cite{eldar2001quantum} to find:

\begin{restatable}{lem}{orbitdisc}\label{lem:orbitdisc}
The optimal average probability of identifying a uniformly random state in $\mathcal{E}_n$ from $t$ copies is
\begin{equation}\label{eq:optprob}
    p_{\mathrm{succ}}^*(t;\mathcal{E}_n)=\frac{1}{|\Upsilon|}\Big(\sum_h \sqrt{p^{(t)}_h}\Big)^2\, .
\end{equation}

\end{restatable}
There is an elementary but surprisingly useful reformulation of Eq.~\eqref{eq:optprob}: let $u$ be the uniform distribution on $\widehat\Upsilon$, i.e. $u_h=1/|\Upsilon|$ for all $h$. Then the right hand side of Eq.~\eqref{eq:optprob} is (the square of) what one might call the \textit{Hellinger overlap} $A(p^{(t)},u)$ between $p^{(t)}$ and $u$:
\begin{equation}
    A(p^{(t)},u):=\sum_{h\in {\widehat\Upsilon}} \sqrt{p^{(t)}_h u_h}=\frac{1}{\sqrt{|\Upsilon|}}\sum_h\sqrt{ p^{(t)}_h}.
\end{equation}
By demonstrating that, for small $t$, the distribution $ p^{(t)}$ is far from uniform (and hence that $A(p^{(t)},u)$ is small) we will therefore prove that the average case learning problem of Lemma~\ref{lem:orbitdisc} cannot be solved, establishing our lower bound.
We will go about this by constructing a real-valued function whose distribution is noticeably different under $p^{(t)}$ and $u$, at least when $t$ is small relative to $n$. The underlying idea will be that   the uniform distribution has zero Fourier coefficient at every nonzero frequency, whereas $p^{(t)}$ has some biases in its Fourier   spectrum; we will align the phases of these biases and add them together so that they reinforce one another into something detectable. A detectable difference implies, we shall see, that the overlap of $p^{(t)}$ and $u$ is small, and hence (by Lemma~\ref{lem:orbitdisc}) that the success probability of learning is small, and hence that one needs a larger value of $t$, establishing the lower bound.\\

First, recall that for a   ``frequency''\footnote{Indeed, note that $\widehat{\widehat{\Upsilon}}\cong\Upsilon$.}  $r=(\alpha,\beta)\in\Upsilon$, and    label   $h=(a,b)\in\widehat\Upsilon$, the character   $h$ evaluated at frequency $r$ is $\chi_{h}(r):=i^{\alpha\cdot a+2\beta\cdot b}$.
Thus, for any probability distribution $q$ on   $\widehat\Upsilon$, our Fourier convention is
\begin{equation}
    \widehat q(r):=\sum_{h\in \widehat\Upsilon}q(h)\chi_{h}(r)=\expect_{h\sim q}\chi_{h}(r)\,.
\end{equation}
It will sometimes be useful to  relate the Fourier coefficients of  $p^{(t)}$ to the original random matrix $\x$, which can be done straightforwardly; indeed, recall that $p^{(t)}$ is, by definition, simply the distribution of the random label $h(\x)$ when $\x\in\mathbb F_2^{n\times t}$ is uniform.  That is, $p^{(t)}$ is the \textit{pushforward} of the uniform distribution under the map $\x\mapsto h(\x)$.  Hence, for every function $F$ of the label $h$,
\begin{equation}
    \expect_{h\sim p^{(t)}}F(h)
    =\expect_{\x\sim\operatorname{Unif}(\mathbb F_2^{n\times t})}F(h(\x))\,;
\end{equation}
we will in   particular be interested in the case  $F=\chi_r$. \\

Let us introduce some more notation. 
We  write the columns of $\x$ as $\x^{(1)},\ldots,\x^{(t)}\in\mathbb F_2^n$, and (as above) its rows as $\x_1,\ldots,\x_n\in\mathbb F_2^t$. Viewing  $r=(\alpha,\beta)$  as the    polynomial
\begin{equation}
r(x)=\sum_j\alpha_jx_j+2\sum_{j<k}\beta_{jk}x_jx_k\pmod 4\,,
\end{equation}
 Eq.~\eqref{eq:h} then gives
\begin{align}
    \chi_{h(\x)}(r)&=i^{\sum_j\alpha_jq_t(\x_j)+2\sum_{j<k}\beta_{jk}(\x_j\cdot \x_k)}\nonumber\\
    &=i^{\sum_{a=1}^t\left(\sum_j\alpha_j\x_j^{(a)} +2\sum_{j<k}\beta_{jk}\x_j^{(a)}\x_k^{(a)}\right)}\nonumber\\
    &=i^{\sum_{a=1}^t r(\x^{(a)})}\,.\label{eq:h39}
\end{align}
At this point, we can write down a  concise expression for the Fourier coefficients $\widehat{p^{(t)}}(r)$. Indeed, letting 
\begin{equation}\label{eq:def-mr}
    m(r):=\expect_{x\sim\operatorname{Unif}(\mathbb F_2^n)}i^{r(x)}\,,
\end{equation}
we have 
\begin{align*}
    \widehat{p^{(t)}}(r)
    &=\expect_{h\sim p^{(t)}}\chi_{h}(r)\\
    &=\expect_{\x^{(1)},\ldots,\x^{(t)}}
        i^{\sum_{a=1}^t r(\x^{(a)})}\\
    &=\expect_{\x^{(1)},\ldots,\x^{(t)}}
        \prod_{a=1}^t i^{r(\x^{(a)})}\\
    &=\prod_{a=1}^t\expect_{\x^{(a)}}i^{r(\x^{(a)})}\\
    &=m(r)^t\,,
\end{align*}
where we have used Eq.~\eqref{eq:h39} and the assumed independence of the   columns of $\x$.\\

We now make bespoke choices of $r$. Let us, for nonzero $ v\in \mbf_2^n$ and
\begin{equation}
    \ell_v(x)=\sum_i v_i x_i+2\sum_{i<j}v_iv_jx_ix_j \mod4\,,\label{eq:pari}
\end{equation}
make the selection
\begin{equation}
    \mathcal{R}:=\{\pm \ell_v: v\in \mbf_2^n\backslash\{0\}\} 
\end{equation}
of $D:=|\mathcal{R}|=2(2^n-1)$ such choices. This choice is slightly more natural than it might appear at first glance; indeed we simply have\footnote{To see this, let $w=\sum_i v_ix_i$; then the right hand side of Eq.~\eqref{eq:pari} is $w+2\binom w2 \ {\rm mod}\ 4=w^2\ {\rm mod}\ 4=w\ {\rm mod}\  2$.} $\ell_v(x)=v\cdot x\ {\rm mod}\ 2$, rewritten in the form of Eq.~\eqref{eq:pari} so as to  correspond to an element of $\Upsilon$.
Now,  for every nonzero $v$, the parity $\ell_v(x)$ is balanced; it equals $0$ on half of $\mbf_2^n$ and $1$ on the other half. Thus,
\begin{align}
    m(\ell_v)&=\mbe_xi^{\ell_v(x)}=\frac{1+i}{2},\\
    m(-\ell_v)&=\mbe_xi^{-\ell_v(x)}=\frac{1-i}{2}\, .
\end{align}
Consequently, $|m(r)|=2^{-1/2}$ for all $r\in \mathcal{R}$.
Fixing $t\geq 1$ and writing $m(r)^t = 2^{-t/2}\omega_{r}$ for some phase $\omega_{r}\in\mbu(1)$,
we can finally construct the long-promised function whose distribution, for small $t$, is noticeably different under $p^{(t)}$ and $u$:
\begin{equation}
    f:\widehat\Upsilon\to\mathbb{R}\,,\quad 
    f(h):=\frac{1}{\sqrt{D}}\sum_{r\in \mathcal{R}}\overline{\omega_r}\chi_{h}(r)\,.
\end{equation}
Note that $f$ is real-valued because the terms indexed by $r$ and $-r$ are complex conjugates, as $\chi_{h}(-r)=\overline{\chi_{h}(r)}$, and $\overline{\omega_{-r}}=\omega_r$. This construction ensures that, averaged over $p^{(t)}$, the factors $\overline{\omega_r}$ cancel with the phases of $m(r)^t$, so that all modes contribute coherently. This leads to a behaviour  that is noticeably different from what is seen when averaging uniformly. Specifically, we find:
\begin{restatable}{lem}{paritywitness}\label{lem:paritywitness}
We have 
$\mathbb{E}_u f = 0,\ \mathbb{E}_u f^2 = 1$, with $u$ the uniform distribution on $\widehat \Upsilon$; furthermore, $\mathbb{E}_{p^{(t)}} f = \sqrt{z}$ and $\mathbb{E}_{p^{(t)}} f^2 \leq 1+z$, where $z:=2(2^n-1)2^{-t}$.
\end{restatable}
Evidently, for $t\ll n$ we have $z\gg 0$, so that Lemma~\ref{lem:paritywitness} detects a strong difference between $p^{(t)}$ and the uniform distribution. This yields   a fundamental obstruction to their Hellinger overlaps being large:
\begin{restatable}{lem}{affinitywitness}\label{lem:affinitywitness}
Let $p$ and $u$ be probability distributions on the same finite set. If a real-valued function $f$ satisfies the conditions of Lemma~\ref{lem:paritywitness}, then
\begin{equation}
    1-A(p,u)^2\geq \frac{1}{8}\min\{1,z\}.
\end{equation}
\end{restatable}
Combining   Lemmas \ref{lem:orbitdisc}, \ref{lem:paritywitness}, and \ref{lem:affinitywitness}, we find that for every $n, t\geq 1$:
\begin{equation*}
    P_\text{err}^*(t,\mathcal{E}_n) = 1-A(p^{(t)},u)^2\geq \frac{1}{8}\min\{1, 2^{n-t}\}.
\end{equation*}

Now, to connect this error probability to the full learning problem, suppose that a $t$-copy learner identifies every state in $\text{Stab}_n$ with failure probability at most $\delta$. Then since $\mathcal{E}_n\subseteq \text{Stab}_n$, its average failure probability specifically on the uniform ensemble $\mathcal{E}_n$ is also at most $\delta$. Thus,
\begin{equation} \label{eq:deltageq}
    \delta\geq P_{\text{err}}^*(t; \mathcal{E}_n)\geq \frac{1}{8}\min\{1, 2^{n-t}\}.
\end{equation}
Let us assume $0<\delta< 1/8$. Then, the minimum in Eq.~\eqref{eq:deltageq} cannot be $1$, as that would imply $\delta\geq 1/8$. Thus $\delta \geq \frac{1}{8}\cdot2^{n-t}=2^{n-t-3}$, whence
\begin{align*}
    t&\geq n+\lceil\log_2(1/\delta)\rceil-3\,,
\end{align*}
which is the lower bound in Theorem \ref{thm:main}.

\section{Discussion}\label{sec:discussion}
In this work we have answered a natural question in quantum learning theory, establishing the precise sample-complexity of learning an unknown stabilizer state. It is interesting to find that the exact lower bounds can be asymptotically saturated by a polynomial time quantum algorithm. As far as seeking to exactly optimise resource costs goes, however, there are a few questions yet unanswered. Indeed, the main limitation of our protocol is the use of a coherent collective measurement on all input copies. It would be interesting to try to obtain the same coefficient-one sample-complexity with measurements restricted to a fixed number of copies at once. It is also open whether   our circuit depths and the  number of ancillae employed in the isotypic compression circuit are optimal, or can be further reduced.  \\

An immediate application of our results is to the learning of  an unknown $n$-qubit Clifford unitary $C$. Applying $C$ to one half  of a $2n$-qubit Bell pair (that is, producing its Choi state~\cite{watrous2018thetheory}) creates a $2n$-qubit stabilizer state (note that the Bell state is itself a stabilizer state). By our results, we can learn this Choi state using $2n+\lceil\log_2(1/\delta)\rceil+4$ such queries, from which $C$ itself can then be readily recovered (up to global phase). For a given failure probability $\delta\in\Omega(1/{\rm poly}\, n)$ this appears to improve upon known sample-complexities: Low~\cite{low2009learning} gives an algorithm which requires both $2n+1$ calls of $C$, and $2n$ calls of $C\ad$;  more recently Skaras and Ginsparg \cite{skaras2026process} gave a forward-only algorithm which requires $4n+3$ calls to $C$. 
Although we show in Appendix~\ref{sec:cliff} that our $n$-dependence is optimal for Clifford learning, the question of the $\delta$-dependence is a little less clear. Indeed, and although learning   an element of a set of nonorthogonal states cannot be done with zero failure probability, Clifford unitaries \textit{can} be   learnt exactly from a finite number of queries~\cite{low2009learning,skaras2026process}. Indeed, for $\delta\in\mco(1/\exp n)$, our sample complexities can rise above $4n$, thereby becoming worse than those of Low~\cite{low2009learning}.  So, in the Clifford case there in principle remains the possibility of improving the $\delta$-dependence, while retaining our optimal $n$-dependence. This is in contrast to our results for learning general stabilizer states (i.e., that are not promised to be the Choi state of a Clifford), where our $\delta$-dependence is optimal.  \\

A second application is to the recently studied problem of
the  sample-complexity of (approximately) \textit{cloning}   an unknown stabilizer state~\cite{bansal2026cloning}. It is not \emph{a priori} clear that cloning and learning should have the same sample-complexities; indeed, perhaps one can (approximately) clone a state without learning it. Exact learning, on the other hand, is clearly sufficient for cloning; our algorithm therefore gives an efficient construction for stabilizer state cloning with an explicitly known scaling coefficient. \\

Finally, one could explore generalisations of our strategy to further groups and representations. For example, this might involve choosing an abelian subgroup $H$ of some target group $G$, and considering a two-step  learning procedure in which one attempts to find an element in $G/H$ which identifies the unknown element up to a remaining degree of freedom on $\mbc[H]$, which is then recovered by a Fourier transform. Naturally, it would be very interesting to both find further examples of groups for which this is efficient, and a  general theory of the conditions     which allow for such efficiency.\\

\begin{acknowledgements}
RC, ML and MW acknowledge support by the Laboratory Directed Research and Development (LDRD) program of Los Alamos National Laboratory (LANL) under project number 20260043DR, and by LANL’s ASC Beyond Moore’s Law project. This work was also supported by the Quantum Science Center (QSC), a National Quantum Information Science Research Center of the U.S. Department of Energy (DOE). 
Part of this work was done while MCC was visiting LANL as a lecturer for their 2026 Quantum Computing Summer School. 
The key technical ideas of the proofs of our results are largely due to GPT-5.6. All AI-generated proofs were verified by the human authors, and  in most cases substantially rewritten in order to meet our standards of clarity.
\end{acknowledgements}


\bibliography{quantum}

\onecolumngrid
\appendix
\renewcommand{\theHequation}{\theHsection.\arabic{equation}}

\section{Stabilizer learning via Bell sampling}\label{sec:bell}
In this Appendix we review Montanaro's Bell sampling algorithm~\cite{montanaro2017learning}, as well as  a simple modification of it. We begin by introducing some notation. We let $\sigma_{00}=\id_2$ be the 2-dimensional identity matrix, $\sg_{01}=X,\ \sg_{10}=Z,\ \sg_{11}=ZX$, and for $s\in\mbf_2^{2n}$ define $\sg_s=\sg_{s_1s_2}\ot \sg_{s_3s_4}\ldots \ot \sg_{s_{2n-1}s_{2n}}$. Every $n$-qubit Pauli string is in this fashion identified with an element of $\mbf_2^{2n}$, in such a way that the natural group laws are (projectively) preserved. That is, $\sg_s\sg_t$ is, up to a possible minus sign, the Pauli string represented by $\sg_{s+t}$ (with the addition in $\mbf_2^{2n}$). Now, any stabilizer state $\ket\psi$ is uniquely specified by a commuting subgroup $G$ of Pauli matrices of size $|G|=2^n$ such that, for $P\in G$, $P\ket\psi=\pm \ket\psi$ (and $\braket{\psi|Q|\psi}=0$ for $Q\not\in G$), along with the knowledge of the specific pattern of $\pm$ signs. Back in $\mbf_2^{2n}$, this corresponds to an $n$-dimensional subspace $T\subseteq \mbf_2^{2n}$, spanned by the set of strings whose corresponding Paulis are (up to a possible phase) in $G$. So, one way to approach stabilizer state learning is to first learn $T$ (and then worry about the phases later).  \\

A key observation of Montanaro~\cite{montanaro2017learning} is that  Bell sampling on two copies of $\ket{\psi}$ gives an outcome $r\in \mathbb{F}_2^{2n}$ that is uniformly distributed on an affine coset $a+T$, for some $a$ that depends only on $\ket\psi$. Then, from independent Bell outcomes $r_0, r_1,\dots,r_m$, the differences $v_i=r_i-r_0$ for $1\leq i\leq m$ are independent and uniformly distributed elements of $T$, which can then be recovered whenever we have found enough   vectors so as to span $T$. Having received (say) $m\geq n$ such samples, the probability of their spanning $T$ is given by
\begin{equation}
    \Pr[\text{span}\{v_1,\dots,v_m\}=T]=\prod_{j=0}^{n-1}(1-2^{j-m})\,;
\end{equation}
indeed, these vectors span $T$ if and only if the (uniformly random) matrix in $\mbf_2^{n\times m}$ (choose any basis) formed from   the columns of the $v_i$ is of full row rank. There are $(2^m-1)(2^m-2)\cdots (2^m-2^{n-1})$ such matrices; dividing by $2^{nm}$ gives the stated probability. We can  then bound the failure probability as
\begin{equation}
    \Pr[\text{span}\{v_1,\dots,v_m\}\neq T]=1-\prod_{j=0}^{n-1}(1-2^{j-m})\leq \sum_{j=0}^{n-1} 2^{j-m}\leq 2^{n-m}.
\end{equation}

Thus for fixed failure probability at most $\delta$, we find that $m=n+\lceil \log_2(1/\delta)\rceil$ is sufficient. Since $r_0,\dots,r_m$ are $m+1$ Bell samples, which each use two input copies, this part needs $2(m+1)=2n+2\lceil \log_2(1/\delta)\rceil+2$ copies. To determine the signs of the Pauli expectation values, Ref.~\cite{montanaro2017learning} proposes simply measuring   each of the now-known $n$ basis Paulis on a new copy of $\ket{\psi}$, taking the  total number of required copies to $3n+2\lceil \log_2(1/\delta)\rceil+2$. Finally, Ref.~\cite{montanaro2017learning} (effectively) takes $\delta=2^{-n}$, leading to a total of $5n+2$ copies. 
In the sign-learning step, if one were instead to measure in the simultaneous eigenbasis of the $n$ (commuting) Paulis known to span $T$, one would need only a single measurement, thus saving $n-1$ copies. With this optimisation (at the cost of a more involved measurement), and allowing for an arbitrary failure probability $\delta$, we conclude that  $2n+2\lceil \log_2(1/\delta)\rceil+3$ copies of $\ket\psi$ are sufficient in this setup.

\section{Efficient isotypic compression}\label{sec:fibers}
\noindent
In this somewhat technical appendix we prove Lemma~\ref{lem:unitary}. We begin, however, with a preliminary lemma:

\begin{restatable}{lem}{fiber}\label{lem:fiber}
 Let $p=(\x_1,\dots,\x_{k})\in(\mbf_2^t)^{\times k}$ and $p'=(\x_1',\dots,\x_{k}')\in(\mbf_2^t)^{\times k}$ satisfy $q_t(\x_i)=q_t(\x_i')$ and $\x_i\cdot\x_j=\x_i'\cdot\x_j'$ for all $1\le i< j\le k$, and that  both $\{\mathbf 1_t,\x_1,\ldots,\x_k\}$ and $\{\mathbf 1_t,\x'_1,\ldots,\x'_k\}$ are linearly independent. Then there exists $g\in G_t:= \{g\in GL_t(\mathbb{F}_2):q_t(gx)=q_t(x) \quad \forall x\in \mathbb{F}_2^t\}$ such that $g\mathbf 1_t=\mathbf 1_t$ and $g\x_i=\x'_i
\ \forall\ 1\leq i\leq k$.
\end{restatable}

\begin{proof}
By linear independence,  the linear map 
\begin{equation}
\Phi: {\rm span}\{\mathbf 1_t, \x_1,\dots,\x_k\}\to {\rm span}\{\mathbf 1_t, \x_1',\dots,\x_k'\}    
\end{equation}
 which fixes $\mathbf 1_t$ and maps $\x_j$ to $\x_j'$ is manifestly  well-defined. Let us show that   it is furthermore a $q_t$-isometry, i.e. that $q_t(\Phi(x))=q_t(x)$. We use the identity
 \begin{equation}
q_t\big(\sum_j c_ju_j\big)=\sum_jc_jq_t(u_j)+2\sum_{j<k}c_jc_k(u_j\cdot u_k) \quad (\text{mod } 4)\,,\label{eq:qid}
 \end{equation}
which follows from the readily verified $q_t(a+b)=q_t(a)+q_t(b)+2a\cdot b\  ({\rm mod}\ 4)$ and induction. 
It follows from Eq.~\eqref{eq:qid}  that $q_t$ of any linear combination of   vectors is completely determined by the $q_t$ values of those  vectors, and their pairwise dot products. By the assumptions of the lemma   combined with the simple observations that $q_t(\mathbf 1_t)=t\ ({\rm mod}\ 4)$ and $\mathbf 1_t\cdot \x_j=q_t(\x_j)\ ({\rm mod}\ 2)$, it follows that $\Phi$ is a $q_t$-isometry. For example, 
\begin{align}
q_t(\Phi(\x_1+\x_2)) &=q_t(\Phi(\x_1)+\Phi(\x_2)) \\
 &= q_t(\Phi(\x_1))+q_t(\Phi(\x_2))+2\Phi(\x_1)\cdot\Phi(\x_2) \quad (\text{mod } 4)\\
&=q_t(\x_1')+q_t(\x_2')+2\x_1'\cdot\x_2' \hspace{24.6mm} (\text{mod } 4)\\
&=q_t(\x_1)+q_t(\x_2)+2\x_1\cdot\x_2 \hspace{24.6mm} (\text{mod } 4)\\
&=q_t( \x_1+\x_2)\,.
\end{align}
Now, we would like to show that $\Phi$ (or, more precisely, some extension of $\Phi$ to a map   $\widetilde{\Phi}:\mbf_2^t\to\mbf_2^t$) may be represented by an element of $G_t$. To that end, we invoke an (unnumbered, sadly) theorem of Wood (\cite{wood1993witts}, Section 4), which states that, defining $I(\mathbb{F}_2^t):=\{x:x\cdot x=0\}=\mathbf 1_t^\perp$ and $ I(\mathbb{F}_2^t)^\perp = \text{span}\{\mathbf 1_t\}$, such an extension will exist precisely if (i) the intersection of the domain of $\Phi$ with $I(\mathbb{F}_2^t)^\perp$ equals the intersection of the range of $\Phi$ with $I(\mathbb{F}_2^t)^\perp$, and (ii) $\Phi$ restricted to $I(\mathbb{F}_2^t)^\perp$ is the identity map. This is clearly true in our case, so that such a $g$ exists. 
\end{proof}

Incidentally, it follows fairly quickly from setting $k=n$ in Lemma~\ref{lem:fiber} that each nonempty accepted isotypic $\mathcal{F}_h$ is exactly a single $G_t$-orbit. To see this, we need first to   verify both the preservation by $G_t$ of the binary dot product, and its stabilisation of $\mathbf 1_t$. The former follows quickly from the polarisation identity $q_t(a+b)=q_t(a)+q_t(b)+2a\cdot b\  ({\rm mod}\ 4)$, so that $2a\cdot b  \equiv 2(ga)\cdot (gb)\ ({\rm mod}\ 4)$ for all $a,b$, implying the preservation of the binary dot product (and, therefore, $g^\mst g = \id$). For the latter, note that for all $x$ we have $x\cdot \mathbf 1_t= q_t(x)\ {\rm mod}\,2= q_t(gx)\ {\rm mod}\,2 = (gx)\cdot 1_t$, whence $x^\mst g^\mst \mathbf 1_t = x^\mst \mathbf 1_t$ for all $x$, so that $ g^\mst \mathbf 1_t=\mathbf 1_t$; multiplying by $g$ then establishes the stabilisation of $\mathbf 1_t$.   Thus, each nonempty accepted isotypic $\mathcal{F}_h$ is exactly a single $G_t$-orbit. $G_t$ will return in Appendix~\ref{sec:cliff}.\\

\noindent
We now turn to the proof of Lemma~\ref{lem:unitary}, which we recall to state:
\unitary*
\begin{proof} 
Let us explicitly describe a circuit construction of $U_h$ for any $h$ corresponding to a non-empty fiber (for an empty fiber, we take $U_h=\id$). Recall that each basis tuple in   $\mcf_h$ is represented by an $n\times t$ matrix $\x$, and that we should like $U_h$  to, when acting on $\ket{0^{nt}}$, prepare the uniform  state $\ket{\mcf_h}$, which we recall to be
\begin{equation}
    \ket{\mcf_h}=\frac{1}{\sqrt{|\mcf_h|}}\sum_{\x\in\mcf_h}\ket{\x}.
\end{equation}
Hence, every valid matrix $\x$ appears with the same amplitude $1/\sqrt{|\mcf_h|}$, so we will construct a coherent superposition over such matrices $\x$, sequentially over the rows. Suppose the first $j-1$ rows $p=(\x_1,\dots,\x_{j-1})$ have already been chosen; we will refer to this selection as the prefix. Now, from $\text{affrank}(\x)=n$ (recall Eq.~\eqref{eq:affrank}) it follows that the next allowed row must be linearly independent from the previous rows (and from $\mathbf 1_t:=(1,1, \ldots 1)\in\mbf_2^t$) and satisfy the conditions of being in the fiber, which we summarise as the     condition
\begin{equation}
    \mcc_j(p)=\left\{\x\in\mathbb{F}_2^t:
 \begin{array}{l}
 q_t(\x)=h_j,\quad \x\cdot \x_k=h_{kj}\ (k<j),\\
 \x\notin\text{span}\{\mathbf 1_t,\x_1,\dots,\x_{j-1}\}
 \end{array}\right\}.
\end{equation}
The first important observation is that, for each fixed $h$ and $j$,    every valid $(j-1)$-row prefix has the same number of valid choices for the next row; let us see this. To begin, if we take two prefixes $p=(\x_1,\dots,\x_{j-1})$ and 
$p'=(\x_1',\dots,\x_{j-1}')$, with both prefixes corresponding to the same $h$,   the corresponding rows fulfill the same conditions, namely  that $q_t(\x_k)=q_t(\x_k')=h_k$ and $\x_k\cdot\x_l=\x_k'\cdot\x_l'=h_{kl}$. Thus, the linear map $\phi$ that sends $\mathbf 1_t$ to $\mathbf 1_t$ and sends each $\x_k$ to $\x_k'$ preserves the quadratic and bilinear structure, and by Lemma~\ref{lem:fiber} we can take it (or, more precisely, some extension of it to a map on all of $\mathbb{F}_2^t$) to be represented by a global $q_t$-isometry $g\in G_t = \{g\in GL_t(\mathbb{F}_2):q_t(g\x)=q_t(\x) \quad \forall \x\in \mathbb{F}_2^t\}$.
Now, take $\x\in \mcc_j(p)$. As $g$ preserves both $q_t$ (by definition) and (as previously verified) the binary dot product, we have
\begin{equation}
    g\x\cdot \x_k'=g\x\cdot g\x_k=\x\cdot \x_k=h_{kj}.
\end{equation}
Furthermore, $g$ is invertible, so $\x\notin \text{span}\{\mathbf 1_t, \x_1, \dots,\x_{j-1}\}$ implies that $g\x \notin \text{span}\{\mathbf 1_t, \x_1', \dots,\x_{j-1}'\}$, so $g\x \in \mcc_j(p')$. 
Conversely, for $\x'\in\mcc_j(p')$ we have $g^{-1}\x' \in \mcc_j(p)$. Thus, $\x\mapsto g\x$ is a bijection from $\mcc_j(p)$ to $\mcc_j(p')$, and thus the count $d_j:=|\mcc_j(p)|=|\mcc_j(p')|$
depends on $h$ and $j$, but not the choice of prefix $p$.
Next, let us check that every row $\x\in \mcc_j(p)$ can actually extend to a full valid matrix in $\mcf_h$. Let us fix a complete matrix $(\boldsymbol{z}_1,\dots,\boldsymbol{z}_n)\in \mcf_h$, and some compatible  prefix and valid new row $(\x_1,\dots,\x_{j-1},\x)$. By assumption   this prefix has the same quadratic and bilinear data as the first $j$ rows of the reference matrix; thus the linear map which acts on the span of $\{\mathbf 1_t, \boldsymbol{z}_1,\ldots, \boldsymbol{z}_{j}\}$ as $\mathbf  1_t\mapsto\mathbf  1_t,\  \boldsymbol{z}_k\mapsto \x_k \text{ for } k<j, \  \boldsymbol{z}_j\mapsto\x$ is a  $q_t$-isometry, and by Lemma~\ref{lem:fiber} extends to some $g\in G_t$. If we then apply $g$ to all the remaining rows in $\boldsymbol{z}$, we get the full tuple $(g\boldsymbol{z}_{1},\ldots,g\boldsymbol{z}_{n})= (\x_1,\dots,\x_{j-1},\x,g\boldsymbol{z}_{j+1},\dots,g\boldsymbol{z}_n)$, which   must also be  in $\mcf_h$ (for, as discussed, $G_t$ preserves fibers). Therefore, any locally valid choice we make for the next row $\x$ actually can be extended to a valid full matrix, so we can make local decisions about each row to build the full matrix.\\

Now, at each stage $j$, every valid prefix leads to $d_j$ possibilities for the next row, so that  the size of the fiber satisfies $|\mcf_h|=\prod_{j=1}^n d_j$.
Thus, if each row is prepared with conditional amplitude $1/\sqrt{d_j}$, every state corresponding to a full matrix will have amplitude $\prod_{j=1}^n {1}/{\sqrt{d_j}}={1}/{\sqrt{|\mcf_h|}}$, as desired. Let us make this explicit.  Say we have fixed a prefix $p$ of previously prepared rows of the matrix. In our current row $\x$, let $\s$ denote the prefix of bits that have already been prepared. Let $N_{\s}$ denote the numbers of valid rows which begin with bitstring $\s$. For example, at the beginning of the row, when $\s$ is empty, $N_\emptyset=|\mcc_j(p)|=N_{0}+N_{1}$, as the first bit can only be $0$ or $1$. The basic idea is that at each successive bit after prefix $\s$, the circuit will compute the two counts $N_{\s0}$ and $N_{\s1}$, and apply the $Y$-rotation which implements
\begin{equation}\label{eq:bitrot}
    \ket{0}\mapsto \sqrt{\frac{N_{\s 0}}{N_{\s}}}\ket{0}+\sqrt{\frac{N_{\s 1}}{N_{\s}}}\ket{1}
\end{equation}
on the next qubit of the row register (if $N_{\s}=0$ we make no rotation). To verify that this construction gives us the correct superposition, let us write out the bits in the new row as $\x=(x^{(1)},\dots,x^{(t)})$, where we prepare the bits from left to right. Then for $0\leq \ell\leq t$, let $\s=(s_1,\dots,s_{\ell})\in \mathbb{F}^{\ell}_2$ be a possible prefix for the first $\ell$ bits of $\x$. Then the conditional amplitudes telescope,
\begin{equation}
    \prod_{{\ell}=0}^{t-1}\sqrt{\frac{N_{s_{{\ell}+1}}}{N_{s_{\ell}}}}=\sqrt{\frac{N_{\x}}{N_\emptyset}}=\frac{1}{\sqrt{|\mcc_j(p)|}},
\end{equation}
so that after the $t$ rotations we will have correctly arrived at the desired uniform superposition state.
It remains to actually find the counts $N_{\s}$. We must first find $N_{\s}'$, the number of solutions to $q_t(\x)=h_j$ and $\x\cdot \x_k=h_{kj}$, and then subtract $N_{\s}''$, the number of those solutions which also fulfill $x\in \text{span}\{\mathbf 1_t, \x_1,\dots, \x_{j-1}\}$.
To find $N_{\s}'$ for fixed $\s$, recall that one requirement set by the definition of $\mcc_j$ is that $\x\cdot \x_k=h_{kj}$ for $(k<j)$. For a bit prefix $
\s\in\mathbb F_2^\ell$, let
\begin{equation}
L_{\s}:=\left\{\x\in\mathbb F_2^t:(x^{(1)},\dots,x^{(\ell)})=\boldsymbol{s},\quad \x\cdot \x_k=h_{kj}\ (k<j)\right\}.
\end{equation}
All the constraints defining $L_{\s}$ are affine linear in the unknown vector $\x$, as the previous rows $\x_k$ as well as the $h_{kj}$ are already fixed. If we then use Gaussian elimination on this linear system of constraints, we either get that the system is unsolvable and $N_{\s}'=0$, or Gaussian elimination gives a particular solution $\x_0$ and a basis $\boldsymbol{v}_1,\dots,\boldsymbol{v}_m$ for the homogeneous solution space, meaning $|L_{\s}|=2^m$. Every vector in $L_{\s}$ can then be written uniquely as
\begin{equation}
    \x=\x_0+\sum_{r=1}^m y_r\boldsymbol{v}_r,\qquad y=(y_1,\dots,y_m)\in \mathbb{F}_2^m.
\end{equation}
This parameterisation is injective because the basis vectors are linearly independent; hence each $y\in \mathbb{F}_2^m$ represents exactly one vector $\x\in L_{\s}$, so counting the number of valid $y$ is equivalent to counting the number of valid $\x$.
Next, we introduce the other constraint, that 
\begin{equation}
    q_t\Big(\x_0+\sum_{r=1}^m y_r \boldsymbol{v}_r\Big)=h_j\mod 4.
\end{equation}
Recalling the previously mentioned \textit{polarisation identity} $q_t(\boldsymbol{u}+\boldsymbol{v})=q_t(\boldsymbol{u})+q_t(\boldsymbol{v})+2(\boldsymbol{u}\cdot \boldsymbol{v})\mod 4$ then yields 
\begin{align}
    P(y):=q_t(\x)&=q_t(\x_0)+\sum_{r=1}^m y_rq_t(\boldsymbol{v}_r)+2\sum_{r=1}^m y_r(\x_0\cdot \boldsymbol{v}_r)+2\sum_{r<r'}y_ry_{r'}(\boldsymbol{v}_r\cdot \boldsymbol{v}_{r'})\mod4\\
    &=c+\sum_{r=1}^m a_ry_r+2\sum_{r<r'}b_{rr'}y_ry_{r'}\mod4 \label{py},
\end{align}
where $c=q_t(\x_0),\ a_r=q_t(\boldsymbol{v}_r)+2(\x_0\cdot \boldsymbol{v}_r)\mod 4 $ and $
    b_{rr'}= \boldsymbol{v}_r\cdot \boldsymbol{v}_{r'}$. 
Counting the number of completions is then equivalent to counting $N_{\s}'=\# \{y\in \mathbb{F}_2^m:P(y)=h_j\mod 4\}$. 
To find this count, consider the  ``filter'' $\id_{\{P(y)=h_j\}}=\frac{1}{4}\sum_{\ell=0}^3 i^{\ell(P(y)-h_j)}$, so that
\begin{align}
    N_{\s}'=\frac{1}{4}\sum_{\ell=0}^3 i^{-\ell h_j}\sum_y i^{\ell P(y)}
    = \frac{1}{4}\Big(2^m+(-1)^{h_j}\sum_y (-1)^{P(y)}+2\Re\Big\{i^{-h_j}\sum_y i^{P(y)}\Big\}\Big)\,.
\end{align}
One can compute both of the two sums   efficiently. The first, $\sum_y (-1)^{P(y)}$, only depends on the parity of $P(y)$, which depends only on the parity of $c+\sum_r a_ry_r,$ so there are no quadratic terms. In fact, we have  
\begin{equation}
    \sum_y (-1)^{P(y)} = (-1)^c\sum_y (-1)^{\sum_r a_r y_r}  = (-1)^c\prod_r \big( 1+ (-1)^{a_r}\big)\,,
\end{equation}
so that the sum is $(-1)^c 2^m$ if all $r$ $a_r$ terms are even, and $0$ otherwise.
For the second sum, consider the diagonal unitary $ U_P = i^c\prod_{r}S_r^{a_r}\prod_{r<r'}\text{CZ}_{rr'}^{b_{rr'}}$, which manifestly satisfies  $U_P\ket{y}=i^{P(y)}\ket{y}$. Since $\ket{+}^{\otimes m}=2^{-m/2}\sum_y\ket{y}$,
\begin{equation}
    \bra{+}^{\otimes m}U_P \ket{+}^{\otimes m} = 2^{-m}\sum_{y,z}\braket{z|U_P|y}=2^{-m}\sum_y i^{P(y)}.
\end{equation}
As $U_P$ can be implemented by a Clifford circuit, $\bra{+}^{\otimes m}U_P \ket{+}^{\otimes m}$ is an inner product between two stabilizer states, and can, including its complex phase, be calculated   in time $\mco(t^3)$~\cite{garcia2017geometry,bravyi2016improved}.
We must  also find $N_{\s}''$. The additional linear constraint that $x\in \text{span}\{\mathbf 1_t, \x_1,\dots,\x_{j-1}\}$ can just be added to the set of constraints defining $L_{\s}$, which is still an affine solution space, so that $N''_{\s}$ can   be calculated by a nearly identical procedure to $N_{\s}'$, again employing   Gaussian elimination to solve the modified linear constraints.
So, we can compute $N_{\s 0}=N'_{\s 0}-N_{\s 0}''$ and $N_{\s 1}=N'_{\s 1}-N_{\s 1}''$ in time $\mco(t^3)$.  Now, to perform the   rotations of Eq.~\eqref{eq:bitrot} we need to (reversibly) calculate $\theta_{\boldsymbol{s}}:=\arcsin\sqrt{{N_{\boldsymbol{s}1}}/{N_{\boldsymbol{s}}}}$ (say, to $b$ bits of precision). To begin, ignore the arcsin, and suppose we have determined the first $j-1$ bits of the integer $R= \lfloor 2^b \sqrt{{N_{\boldsymbol{s}1}}/{N_{\boldsymbol{s}}}}\rfloor$.  The $j$\textsuperscript{th} bit will be 1 if and only if $2^b \sqrt{{N_{\boldsymbol{s}1}}/{N_{\boldsymbol{s}}}} \geq R + 2^{b-j}$. Squaring both sides, this is just a matter of comparing two $\mco(t+b)$-bit integers. The squaring   takes (naively) time $\mco((t+b)^2)$, and the whole square root calculation therefore  $\mco(b(t+b)^2)$. To subsequently evaluate the arcsin we can simply take the first $\mco(b)$ terms of its Taylor series; evaluating any such term takes time $\mco(b^2)$ (strictly, to guarantee this convergence we need our angle   to be bounded away from 1; we can guarantee this by checking if $N_{\boldsymbol{s}1}/N_{\boldsymbol{s}}>1/2$ and, if it is, instead calculating $\pi/2 - \arcsin(\sqrt{N_{\boldsymbol{s}0}/N_{\boldsymbol{s}}})$). Putting everything together, we find that we can produce a $b$-bit estimate of $\theta_{\boldsymbol{s}}$ in time $\mco(t^3 + b(t+b)^2+b^3) = \mco((t+b)^3)$; we can make each of these computations reversible using $\mco((t+b)^3)$ ancillae without changing the asymptotic scaling of the depth~\cite{bennett1973logical}. This whole calculation is repeated for all $nt$ bits of $\boldsymbol{X}$, leading to an overall complexity of $\mco(nt(t+b)^3)$.     \\

The remaining question, then, is how many bits of precision we need in these angle calculations. Let $R_1,\dots,R_{nt}$ denote the ideal rotations, and $\widetilde{R}_1,\dots,\widetilde R_{nt}$  their implemented approximations. Given some target label $h$, at the $k$\textsuperscript{th} rotation step we need to implement $W_k=\sum_c \ketbra{c}\ot R_k(2\theta_{h,c,k})$, where the control (computational basis) state $\ket c$ encodes the previously discussed classical information needed to calculate $\theta_{h,c,k}$ (apart from $h$ itself), and $R_k$ denotes a rotation about the $Y$-axis acting on the qubit which is the target of timestep $k$ (to be clear, we do not attempt to synthesise each rotation individually; we instead synthesise rotations by $2^{-j}$ for $1\leq j\leq b$ and then for each target rotation we apply some subset of these rotation gates, as determined by $c$). Now, suppose that $\max_{h,c,k} \|\widetilde R_k(\theta_{h,c,k})-R_k(\theta_{h,c,k})\|_\infty \leq \varepsilon $ (so that also $\max_{k} \|\widetilde W_k-W_k \|_\infty \leq \varepsilon $). The readily verified telescoping identity
 \begin{equation}
   \widetilde W_M\cdots\widetilde W_1-W_M\cdots W_1=
\sum_{k=1}^{M}\widetilde W_M\cdots\widetilde W_{k+1}(\widetilde W_k-W_k)W_{k-1}\cdots W_1  
 \end{equation}
 then yields, with $V_h$ and $\widetilde U_h$ the exact and approximate implementations of the entire unitary,
 \begin{equation}
     \| V_h\ket{0}- \widetilde U_h\ket 0\|_2 \leq \| V_h- \widetilde U_h\|_\infty  \leq\sum_{k=1}^{nt}\left\|\widetilde W_k-W_k\right\|_{\infty} \leq nt\varepsilon
 \end{equation}
Therefore, we can choose $\varepsilon\leq \zeta/(nt)$ to make the final (2-norm) error at most $\zeta$;   this can be synthesised into   depth $\mco(b\log(b/\varepsilon))$~\cite{ross2014optimal,dawson2005solovay}, where the extra factors of $b$ (which we can take to be $\lceil\log(nt/\zeta)\rceil$) come from the fact that each rotation is broken into $b$ rotations by various powers of two, each of which should be implemented to error $\varepsilon/b$.  
Finally, we of course uncompute the ancilla register. Putting everything together, then, our final estimate of the complexity of the isotypic compression step is   $\mco(nt(t+\log(nt/\zeta))^3)$.  
\end{proof}

\section{Proof of the lower bound}\label{sec:lbapp}

\orbitdisc*
\begin{proof}
Recalling   our definition $\mcf_h = \{\x:h(\x)=h\text{ and }\text{affrank}(\x)=n\}$  to exclude non-full rank states, let us introduce the relaxed analogues  $\mcg_h = \{\x:h(\x)=h\}$, so for example that $\ket{\phi_{\a,\b}}^{\otimes t} = \sum_{h\in \widehat{\Upsilon}} \sqrt{p^{(t)}_h} \chi_h((\a,\b)) \ket{\mathcal{G}_h}$. Now, for ${(\a,\b)}\in\Upsilon$, let us introduce the   diagonal unitaries $U_{\a,\b}$, which act as $U_{\a,\b}\ket{\mathcal{G}_h}=\chi_h((\a,\b))\ket{\mathcal{G}_h}$. Note that from the satisfaction by the characters of $\chi_h((\a,\b)+(\a',\b'))=\chi_h((\a,\b))\chi_h((\a',\b'))$, we immediately have   $U_{\a,\b}U_{\a',\b'}=U_{(\a,\b)+(\a',\b')}$. Thus, the family of states is a (\textit{geometrically uniform}, if you like) orbit of the abelian group given by the matrices $U_{\a,\b}$:
\begin{equation}
    \ket{\phi_{\a,\b}}^{\otimes t} = U_{\a,\b}\ket{\phi_{0,0}}^{\otimes t}\,,
\end{equation}
so that they are optimally discriminated by a \textit{pretty good   measurement}~\cite{eldar2001quantum}. Let us go through the details. \\

\noindent
First, suppose that $\{N_{\a,\b}\}_{ ({\a,\b})\in \Upsilon}$ is any POVM intended to guess the label $({\a,\b})$, with average success probability 
\begin{equation}
    p_{\mathrm{succ}}(N)=\frac{1}{|\Upsilon|}\sum_{(\a,\b)\in \Upsilon}\bra{\phi_{\a,\b}}^{\otimes t} N_{\a,\b} \ket{\phi_{\a,\b}}^{\otimes t}\,.
\end{equation}
In a well-known trick, we can define another measurement by averaging $N$ over $\Upsilon$:
\begin{equation}
    M_{\a,\b}:=\frac{1}{|\Upsilon|}\sum_{(\a',\b')\in \Upsilon}U_{\a',\b'}N_{(\a,\b)-(\a',\b')}U_{\a',\b'}^\dagger.
\end{equation}
Note that this is indeed a POVM; its elements are PSD as they are sums of PSD operators, and 
\begin{equation*}\phantom{.}\hspace{-6mm}
    \sum_{\a,\b} M_{\a,\b} = \frac{1}{|\Upsilon|}\sum_{\substack{\a,\b\\\a',\b'}}U_{\a',\b'}N_{(\a,\b)-(\a',\b')}U_{\a',\b'}\ad
    = \frac{1}{|\Upsilon|} \sum_{\a',\b'} U_{\a',\b'} \sum_{\a,\b} N_{({\a,\b})-({\a',\b'})}U_{\a',\b'}^\dagger
    = \frac{1}{|\Upsilon|}\sum_{\a',\b'} U_{\a',\b'}\id U_{\a',\b'}^\dagger
    = \id\,.
\end{equation*}
Now, we can calculate:
\begin{align}
    p_{\mathrm{succ}}(M)&=\frac{1}{|\Upsilon|^2}\sum_{\a,\b,\a',\b'}\bra{\phi_{\a,\b}}^{\otimes t}U_{\a',\b'} N_{(\a,\b)-(\a',\b')}U_{\a',\b'}^\dagger \ket{\phi_{\a,\b}}^{\otimes t}\\
    &= \frac{1}{|\Upsilon|^2}\sum_{\a,\b,\a',\b'}\bra{\phi_{(\a,\b)-(\a',\b')}}^{\otimes t}N_{(\a,\b)-(\a',\b')}\ket{\phi_{(\a,\b)-(\a',\b')}}^{\otimes t}\\
    &= \frac{1}{|\Upsilon|}\sum_{\a'',\b''} \bra{\phi_{\a'',\b''}}^{\otimes t}N_{\a'',\b''} \ket{\phi_{\a'',\b''}}^{\otimes t}\\
    &= p_{\mathrm{succ}}(N).
\end{align}
So, we can without loss of generality take the optimal measurement to be symmetrised in the above sense.
Therefore, we can consider just these  measurements $M_{\a,\b}$ when calculating for average success. 

\noindent
Next, note that the symmetrised measurement satisfies $M_{({\a,\b})+({\a',\b'})}=U_{\a',\b'}M_{\a,\b}U_{\a',\b'}^\dagger,$ and in particular $M_{\a,\b}=U_{\a,\b}M_{0,0}U_{\a,\b}^\dagger$. This yields:
\begin{equation}
    \bra{\phi_{\a,\b}}^{\otimes t} M_{\a,\b} \ket{\phi_{\a,\b}}^{\otimes t} = \bra{\phi_{0,0}}^{\otimes t}U_{\a,\b}^\dagger U_{\a,\b} M_{0,0} U_{\a,\b}^\dagger U_{\a,\b}\ket{\phi_{0,0}}^{\otimes t} = \bra{\phi_{0,0}}^{\otimes t}M_{0,0}\ket{\phi_{0,0}}^{\otimes t}.
\end{equation}
Thus, every state has the same conditional success probability. Furthermore, for  $h,k\in\widehat\Upsilon$:
\begin{equation}
    \bra{\mathcal{G}_h}M_{\a,\b}\ket{\mathcal{G}_k}=\bra{\mathcal{G}_h}U_{\a,\b} M_{0,0}U_{\a,\b}^\dagger \ket{\mathcal{G}_k}=\chi_h((\a,\b))\overline{\chi_k((\a,\b))}\bra{\mathcal{G}_h}M_{0,0}\ket{\mathcal{G}_k}.
\end{equation}
Now, character orthogonality says that $\sum_{(\a,\b)\in \Upsilon} \chi_h((\a,\b))\overline{\chi_k((\a,\b))}=|\Upsilon|1_{h=k}$, so that 
\begin{equation}
    \bra{\mathcal{G}_h} \sum_{(\a,\b)\in \Upsilon}M_{\a,\b} \ket{\mathcal{G}_k}=\bra{\mathcal{G}_h}M_{0,0}\ket{\mathcal{G}_k}\sum_{\a,\b} \chi_h((\a,\b))\overline{\chi_k((\a,\b))}=|\Upsilon|1_{h=k}\bra{\mathcal{G}_h}M_{0,0}\ket{\mathcal{G}_k};
\end{equation}
from the normalisation $\sum_{\a,\b} M_{\a,\b}=I$ we conclude that for all $h$ we have $\bra{\mathcal{G}_h}M_{0,0}\ket{\mathcal{G}_h}=1/|\Upsilon|$.
Turning to  the off-diagonal elements,  $M_{0,0}$ being positive semidefinite forces its $2\times 2$  submatrix on $\ket{\mathcal{G}_h}, \ket{\mathcal{G}_k}$ to  be positive semidefinite:
\begin{equation}
    \begin{pmatrix} 
        {1}/{|\Upsilon|}& (M_{0,0})_{hk}\\
        \overline{(M_{0,0})_{hk}}& 1/{|\Upsilon|}
    \end{pmatrix}\succeq0,
\end{equation}
so that its determinant ${1}/{|\Upsilon|^2}-|(M_{0,0})_{hk}|^2$ is nonnegative, whence $|(M_{0,0})_{hk}|\leq {1}/{|\Upsilon|}$. Putting some of the recently discussed points together, we see that the average probability is given by
\begin{align*}
    p_{\mathrm{succ}}(M) = \bra{\phi_{0,0}}^{\otimes t}M_{0,0}\ket{\phi_{0,0}}^{\otimes t}&=\sum_{h,k}\sqrt{p^{(t)}_hp^{(t)}_k}\bra{\mathcal{G}_h}M_{0,0}\ket{\mathcal{G}_k}\\
    &\leq \sum_{h,k}\sqrt{p^{(t)}_hp^{(t)}_k}\lvert\bra{\mathcal{G}_h}M_{0,0}\ket{\mathcal{G}_k}\rvert\\
    &\leq \frac{1}{|\Upsilon|}\sum_{h,k}\sqrt{p^{(t)}_hp^{(t)}_k}\\
    &=\frac{1}{|\Upsilon|}\Big(\sum_h\sqrt{p^{(t)}_h}\Big)^2\, ,
\end{align*}
where we have without loss of generality evaluated on $\ket{\phi_{0,0}}\tt$.  So, no measurement can enjoy an average success   $p_{\mathrm{succ}}>(\sum_h\sqrt{p^{(t)}_h})^2/|\Upsilon|$.\\

\noindent
Let us now see that this   bound is attainable.
Define the unnormalised vector $\ket{\mu}=\frac{1}{\sqrt{|\Upsilon|}}\sum_{h:p^{(t)}_h>0}\ket{\mathcal{G}_h}$, and let $M_{0,0}=\ketbra{\mu}$. Then for  $h,k$ with $p^{(t)}_h,p^{(t)}_k>0$, we find $\bra{\mathcal{G}_h}M_{0,0}\ket{\mathcal{G}_k}=\braket{\mcg_h|\mu}\braket{\mu|\mcg_k}={1}/{|\Upsilon|}$.
Taking the other measurements to be $M_{\a,\b}=U_{\a,\b}M_{0,0}U_{\a,\b}^\dagger$, we obtain by character orthogonality   a valid POVM on the subspace spanned by  the  $\ket{\mcg_h}$ with $p^{(t)}_h>0$ (including thereby the support of the $t$-th tensor powers of the states in $\mce_n$).
Then the success probability for $(\a,\b)=(0,0)$ is
\begin{equation}
    \bra{\phi_{0,0}}^{\ot t} M_{0,0}\ket{\phi_{0,0}}^{\ot t}=\lvert\braket{\mu|\phi_{0,0}^{\ot t}}\rvert^2=\frac{1}{|\Upsilon|}\Big(\sum_h \sqrt{p^{(t)}_h}\Big)^2\,,
\end{equation}
saturating   the upper bound, and thus giving the optimal success probability. 

\end{proof}

\paritywitness*
\begin{proof}
Recall that 
\begin{equation}
    f(h)=\frac{1}{\sqrt{D}}\sum_{r\in \mathcal{R}}\overline{\omega_r}\chi_{h}(r)\,,
\end{equation}
where $\omega_r$ is the phase of $m(r)^t$. 
First, for every uniformly sampled $h\in \widehat\Upsilon$, character orthogonality says that $\mbe_u \chi_{h}(r)=\d_{r,0}$; but every $r\in \mathcal{R}$ is nonzero, so that $\mbe_u f = \frac{1}{\sqrt{D}}\sum_{r\in \mathcal{R}} \overline{\omega_r}\mbe_u \chi_{h}(r) =0$. 
For the second claim, we can expand:
\begin{equation}
    f(h)^2= \frac{1}{D}\sum_{r,s\in \mathcal{R}} \overline{\omega_r}\,\overline{\omega_s}\chi_h(r+s).
\end{equation}
When we average over $u$, every term again contributes $0$ by character orthogonality except when $r+s=0$. Therefore,
\begin{align}
    \mbe_u f^2 = \frac{1}{D}\sum_{r\in \mathcal{R}} \overline{\omega_r}\,\overline{\omega_{-r}}= \frac{1}{D}\sum_{r\in \mathcal{R}}\overline{\omega_r}\,\omega_r= \frac{1}{D}\sum_{r\in \mathcal{R}}1 = 1.
\end{align}
For the third claim, we use the fact that $\mbe_{h\sim p^{(t)}} \chi_{h}(r)=m(r)^t$:
\begin{equation}
    \mbe_{p^{(t)}} f = \frac{1}{\sqrt{D}}\sum_{r\in \mathcal{R}}\overline{\omega_r}m(r)^t= \frac{1}{\sqrt{D}}\sum_{r\in \mathcal{R}}\overline{\omega_r}(2^{-t/2}\omega_r)= \frac{1}{\sqrt{D}}\sum_{r\in \mathcal{R}}2^{-t/2}= \sqrt{D2^{-t}}= \sqrt{z}.
\end{equation}
For the final claim, we again expand the square, obtaining $\mbe_{p^{(t)}} f^2 =\frac{1}{D}\sum_{r,s\in \mathcal{R}}\overline{\omega_r}\,\overline{\omega_s}m(r+s)^t$.
Here, we have three possible cases: one case when $r=\sigma \ell_v$, $s=\tau \ell_w$, where $v\neq w$ and $\sigma,\tau\in \{\pm1\}$, and the other two when $v=w$ so $r=s$ or $r=-s$.
We start with the case when $s=-r$. Here, $m(r+s)^t=m(0)^t=1$. There are $D$ such pairs, so in total this contributes $1$ to the average.
Next, for the case when $s=r=\pm \ell_v$, we have $r+s=\pm 2 \ell_v$. Then, as the parity function is balanced,
\begin{equation}
    m(\pm 2 \ell_v)= \mbe_x i^{\pm 2\ell_v(x)}= \mbe_x(-1)^{\ell_v(x)}=0,
\end{equation}
so this case contributes nothing.
Finally, when $r=\sigma \ell_v$, $s=\tau \ell_w$, where $v\neq w$ and $\sigma,\tau\in \{\pm1\}$, we have:
\begin{equation}
    |m(r+s)|=|\mbe_x i^{\sigma \ell_v(x)+\tau \ell_w (x)}|= |\mbe_x i^{\sigma \ell_v(x)}||\mbe_x i^{\tau \ell_w(x)}|= 2^{-1/2}\cdot 2^{-1/2}=\frac{1}{2}.
\end{equation}
Then, $|m(r+s)^t|=2^{-t}$. There are $D(D-2)$ choices here because for each of the $D$ choices for $r$, we count all the $s$ except when $s=\pm r$. Thus in total, this case contributes at most  $(D-2)2^{-t}$.
Then in total,
\begin{equation}
    \mbe_{p^{(t)}} f^2 \leq 1+(D-2)2^{-t}\leq 1+D2^{-t}= 1+z.
\end{equation}
\end{proof}

\affinitywitness*
\begin{proof}
We begin by recalling that the presupposed conditions  are: $\mathbb{E}_u f = 0,\ \mathbb{E}_u f^2 = 1$;  $\mathbb{E}_{p} f = \sqrt{z},$ and $\mathbb{E}_{p} f^2 \leq 1+z$, where $z:=2(2^n-1)2^{-t}>0$ for $n>1$. These suppositions imply
\begin{align}
\Big(\sum_h f(h)[p_h-u_h]\Big)^2 =(\mbe_p[f]-\mbe_u[f])^2 = (\sqrt z-0)^2 = z\,;
\end{align}
Cauchy-Schwarz then yields
\begin{align}
    z &= \Big(\sum_h f(h)[p_h-u_h]\Big)^2\\
    &= \Big(\sum_h f(h)(\sqrt{p_h}+\sqrt{u_h})(\sqrt{p_h}-\sqrt{u_h})\Big)^2\\
    &\leq \Big(\sum_h(\sqrt{p_h}-\sqrt{u_h})^2\Big)\Big(\sum_h f(h)^2(\sqrt{p_h}+\sqrt{u_h})^2\Big).
\end{align}
For the first parenthetical, notice that 
\begin{align}
    \Big(\sum_h(\sqrt{p_h}-\sqrt{u_h})^2\Big) &= \sum_h p_h-2\sum_h\sqrt{p_h u_h}+\sum_h u_h\\
    &=2(1-\sum_h\sqrt{p_hu_h})\\
    &=2(1-A(p,u)).
\end{align}
For the second, we can again expand the squares and dutifully apply Cauchy-Schwarz:
\begin{align}
    \sum_h f(h)^2(\sqrt{p_h}+\sqrt{u_h})^2 &= \sum_h f(h)^2p_h+2\sum_h f(h)^2\sqrt{p_hu_h}+\sum_h f(h)^2u_h\\
    &\leq \sum_h f(h)^2 p_h+2\sqrt{\Big(\sum_h f(h)^2 p_h\Big)\Big(\sum_h f(h)^2 u_h\Big)} + \sum_h f(h)^2 u_h\\
    &= \left(\sqrt{\sum_h f(h)^2 p_h}+\sqrt{\sum_h f(h)^2 u_h}\right)^2\\
    &\leq (\sqrt{1+z}+1)^2.
\end{align}
Putting the pieces together, we have
\begin{align}
    z\leq 2(1-A(p,u))(\sqrt{1+z}+1)^2\,;
\end{align}
a little rearranging then gives
\begin{align}
    1-A(p,u)&\geq \frac{z}{2(\sqrt{1+z}+1)^2}\label{eq:xkz}\\
    &= \frac{z(\sqrt{1+z}-1)}{2(\sqrt{1+z}+1)^2(\sqrt{1+z}-1)}\\
    &= \frac{\sqrt{1+z}-1}{2(\sqrt{1+z}+1)}\\
    &= \frac{\sqrt{1+z}+1}{2(\sqrt{1+z}+1)}-\frac{2}{2(\sqrt{1+z}+1)}\\
    &= \frac{1}{2}-\frac{1}{\sqrt{1+z}+1}\label{xinc}.
\end{align}

We now have a lower bound on $1-A(p,u)$, but recall $P_{\mathrm{succ}}^*=A(p,u)^2$, so we want to manipulate our bound to be on $A(p,u)^2$ instead. The simplest way comes from splitting into cases $0\leq z\leq 1$ and $z > 1$. For readability in the following details, we will denote $A(p,u)$ with just $A$, and let $k(z)=z/(2(\sqrt{1+z}+1)^2)$, and $x=1-A$. Note that $0\leq A\leq 1$, so that also $0\leq x\leq 1$, and $1-A^2=2x-x^2$. Now, as the function $2x-x^2$ is increasing on $x\in [0,1]$, and as by Eq.~\eqref{eq:xkz} $x\geq k(z)$, we conclude
\begin{equation}
    1-A^2\geq 2k(z)-k(z)^2\,.
\end{equation}
Now, for the case when $0\leq z\leq 1$, note that $2\leq \sqrt{1+z}+1<\frac{5}{2}$, so 
\begin{equation}
    \frac{2z}{25}\leq k(z)\leq \frac{z}{8}\leq \frac{1}{8}.
\end{equation}
Then,
\begin{align}
    2k(z)-k(z)^2&=k(z)(2-k(z))\geq \frac{2z}{25}(2-\frac{1}{8})= \frac{3z}{20}>\frac{z}{8}.
\end{align}
For the case $z\geq1$, note that the form of $k(z)$ written exposed in Eq.~\eqref{xinc} shows that $k(z)$ is increasing. Therefore, for $z>1$, we can write:
\begin{equation}
    2k(z)-k(z)^2\geq 2k(1)-k(1)^2\geq \frac{3}{20}>\frac{1}{8}.
\end{equation}
Thus, when $0\leq z\leq 1$, $1-A^2>\frac{z}{8}$, and when $z>1$, $1-A^2>\frac{1}{8}$. Together, $1-A^2\geq \frac{1}{8}\min\{1,z\}$.
\end{proof}

\section{On the optimality of Clifford learning}\label{sec:cliff}
\noindent
In this appendix we argue that any algorithm for learning, with failure probability at most 1/8, an unknown $n$-qubit Clifford unitary $C$ from forward query access (i.e., without access to $C\ad$)   requires at least $2n$ queries. First, by the principle of deferred measurement~\cite{nielsen2000quantum}, we can associate to any scheme involving mid-circuit measurements (and subsequent conditioning on the results of those measurements) an equivalent procedure in which all of the measuring happens at the end. We can therefore without loss of generality model any $t$-query learning procedure as performing a measurement on the state~\cite{chiribella2013identification}
\begin{equation}\label{eq:pure_circ}
 \ket{\Psi_C}=W_t(C\otimes\id)W_{t-1}\cdots W_1(C\otimes\id)W_0\ket0\,,
\end{equation}
where we have allowed the ($C$-independent) isometries $W_i$ to act   also on arbitrary ancilla spaces. Now, inspecting Eq.~\eqref{eq:pure_circ}, we see that there exists some ($C$-independent) linear map $A_t$ such that $\ket{\Psi_C}=A_t (\sdket{C}^{\ot t})$, where $\sdket{C}=(C\ot\id)\ket\Phi$ is the Choi state of $C$. We emphasise that we do not care if $A_t$ is a CPTP map; it is simply a temporary tool that allows us to notice that
\begin{equation}\label{eq:span_dims}
 \dim{\rm span}\{\ket{\Psi_{C}}:C\in \mathsf{Cl}_n\}=
 \dim{\rm span}\{A_t\sdket{C}^{\otimes t}:C\in \mathsf{Cl}_n\}\leq\dim{\rm span}\{\sdket{C}^{\otimes t}:C\in \mathsf{Cl}_n\}\,.
\end{equation}
The usefulness of this observation comes from the relation between the difficulty of (uniform) state discrimination on   some set and the dimension of the space spanned by the elements of that set; indeed, let $\{\ket{\psi_x}:x\in \mcx\}$ be pure states, all contained in some subspace $W$, with  $\Pi_W$ be the orthogonal projector onto $W$, and let $\{M_x\}_{x\in \mcx}$ be the POVM elements corresponding to  the guesses of some procedure for learning states from $\mcx$.  Since $\ketbra{\psi_x}\leq\Pi_W$, we have that the success probability (under the uniform prior) is bounded as
\begin{equation}\label{eq:orbprob}
 p_{\mathrm{av}}=\frac{1}{|\mcx|}\sum_x\tr\left(M_x\ketbra{\psi_x}{\psi_x}\right)\leq\frac{1}{|\mcx|}\sum_x\tr(M_x\Pi_W)\leq\frac{\tr\Pi_W}{|\mcx|}=\frac {\dim W}{|\mcx|}\,,
\end{equation}
where we have used $\sum_x M_x\leq \id$. 
In our present circumstances, this translates into a bound 
\begin{equation}\label{eq:pavbound}
p_{\rm av}(t)\leq {\dim V_{2n,t}}/{|\mathsf{Cl}_n|}\,,    
\end{equation}
where the dimension of $V_{2n,t}:={\rm span}\{\ket S^{\otimes t}:S\in{\rm Stab}_{2n}\} $ 
is the quantity to which we now turn (note that the Choi states of $n$-qubit Cliffords are a strict subset of the total set of  $2n$-qubit stabilizer states, but considering as we are in Eq.~\eqref{eq:pavbound} anyway  a lower bound  allows us to pass to this entire set, which we can analyse somewhat more readily). Now,
let us recall from Appendix~\ref{sec:fibers} the group
\begin{equation}
G_t= \{
g\in GL_t(\mathbb{F}_2)
:
q_t(gx)=q_t(x)
\quad
\forall x\in\mathbb{F}_2^t
\}.
\end{equation}
Here, also recall $q_t(x):=|x|\pmod 4$ denotes the Hamming weight of $x\in\mathbb{F}_2^t$ reduced modulo four, and thus $G_t$ is precisely the subgroup of the invertible binary matrices that preserves Hamming weight modulo four. The reason mod $4$ appears rather than mod $2$ is ultimately the Clifford/stabilizer phases $i^{|x|}$: knowing $|x|\pmod4$ determines $i^{|x|}$. We will need a few facts about $G_t$. The first is that, with $g\in G_t$ acting on $\mbf_2^t$ in its defining representation, and acting diagonally on $(\mbf_2^{2n})^t\cong (\mbf_2^t)^{2n}$ via the natural representation $R_{2n}(g)=g^{\ot {2n}}$, we have $R_{2n}(g)\ket S^{\otimes t}=\ket S^{\otimes t}$ for all $g\in G_t$, and all $2n$-qubit stabilizer states $\ket S$~\cite{gross2021schur}. The fact we need is that,  when $t=2m+1$ is odd, one can show that we have $G_t\cong O^{\varepsilon_t}(2m,Q)$, where $Q=|x|/2\mod 2 $ is the $G_t$-preserved quadratic form on $E$, and $\varepsilon_t$ is 1 for $t\cong 1,7\mod 8$, and $-1$ for $t\cong 3,5\mod 8$~\cite{wood1993witts}. We don't particularly need to understand anything about the groups $O^{\pm}(2m,Q)$, except for the fact that their order can be calculated~\cite{wood1993witts,wilson2009finite}, leading to (for odd $t\geq 3$)
\begin{equation}\label{eq:ogt}
|G_t|=2^{(t^2-3t+4)/2}(1-\varepsilon_t2^{-m})\prod_{j=1}^{m-1}(1-2^{-2j})\,.
\end{equation}
This is relevant because the stabilisation of the elements of ${\rm Stab}_{2n}$ by $R_{2n}$ implies, lying therefore  as they do in the trivial isotypic of $R_{2n}$, that
\begin{equation}
   \dim  V_{2n,t}\leq \tr[\frac{1}{|G_t|}\sum_{g\in G_t}R_{2n}(g)]=\frac{1}{|G_t|}\sum_{g\in G_t}2^{2n\dim \ker (g-\id_t)}\,.
\end{equation}
So, we find ourselves interested in the question of how many $g$ can have (say) $\rank(g-\id_t)=r$. Fixing such a $g$, let  $W=\ker(g-\id_t)$. Now, $ {\rm im}(g-\id_t)\subseteq W^\perp$; indeed, note that for $w\in W$ and any $v\in\mbf_2^t$ we have 
\begin{equation}
((g-\id_t)v)\cdot w = (gv)\cdot w - v\cdot w = v\cdot (g^{-1}w) - v\cdot w = v\cdot w - v\cdot w =0\,,
\end{equation}
so that $(g-\id_t)$ induces a well-defined map  $ \mbf_2^t/W\longrightarrow W^\perp$. With both the domain and range having dimension $r$, there are, for a fixed $W$, at most $2^{r^2}$ such maps,  and thus at most that many possible $g$. On the other hand, we find in the proof of Lemma~\ref{lem:fourier} that there are at most $4\cdot 2^{r(t-r)}$ such subspaces $W$, so that (for $2n>t$) we have
\begin{align}
    \dim V_{2n,t}&\leq \frac{1}{|G_t|}\sum_{g\in G_t}2^{2n\dim \ker (g-\id_t)}\\
     &\leq\frac1{|G_t|}\left(
 2^{2nt}+4\sum_{r=1}^t2^{rt+2n(t-r)}\right)\\
 &\leq\frac{2^{2nt}}{|G_t|}
 \left(1+4\sum_{r=1}^{\infty}2^{-r(2n-t)}\right)\\
 &=\frac{2^{2nt}}{|G_t|}
 \left(1+\frac4{2^{2n-t}-1}\right)\,.
\end{align}
Additionally, recall $|\mathsf{Cl}_{n}|=2^{n^2+2n}\prod_{j=1}^n(4^j-1)$; or,  introducing $ Q_k:=\prod_{j=1}^k(1-4^{-j})$ to simplify some upcoming notation, $|\mathsf{Cl}_{n}|=2^{2n^2+3n}Q_n$. \\

At this point, we are ready to put the various pieces together to show that at least $2n$ queries are required in order to learn an unknown Clifford with  failure probability at most $1/8$. To do this, we will set $t=2n-1$, and show that the success probability is in that instance strictly less than $7/8$. This choice of $t$ (and its subsequent implication that $m=(t-1)/2=n-1$) leads from Eq.~\eqref{eq:ogt} to
\begin{equation}\label{eq:ogt2}
|G_{2n-1}|=2^{2n^2-5n+4}(1-\varepsilon_{2n-1}2^{-(n-1)})Q_{n-2}\,,
\end{equation}
whence
\begin{align}
    p_{\rm av}(2n-1)&\leq \frac{\dim V_{2n,2n-1}}{|\mathsf{Cl}_n|}\\
    &\leq \frac{2^{2n(2n-1)}}{|G_{2n-1}|\,|\mathsf{Cl}_{n}|}\left(1+\frac4{2-1}\right) \\
    &= 5\frac{2^{2n(2n-1)}}{(2^{2n^2-5n+4}(1-\varepsilon_{2n-1}2^{-(n-1)})Q_{n-2}) (2^{2n^2+3n}Q_n)} \\
    &= \frac{5}{16(1-\varepsilon_{2n-1}2^{-(n-1)})Q_{n-2}Q_n}\label{eq:denfun} \\
    &\leq \frac{45}{64(1-\varepsilon_{2n-1}2^{-(n-1)})} \,,\label{eq:pav-bound}
\end{align}
where we have used that (for all $k$)  we have
\begin{equation}
    Q_k=\prod_{j=1}^k(1-4^{-j})\geq\prod_{j=1}^\infty(1-4^{-j})\geq 1-\sum_{j=1}^\infty4^{-j}=\frac23\,.
\end{equation}
We want to show that Eq.~\eqref{eq:pav-bound} is less than $7/8$ for all $n$. Now, if $n\geq 4$,  we have $1-\varepsilon_{2n-1}2^{-(n-1)}\geq 7/8$, so that $p_{\rm av}(2n-1)\leq 45/56<7/8$, establishing the result for $n\geq 4$. If $n=3$, it is easy to see that the factors after the 16 in the denominator of Eq.~\eqref{eq:denfun} are respectively bounded below by $3/4, 3/4$, and $2/3$, leading to $p_{\rm av}(5)\leq 5/(16(3/4)(3/4)(2/3))=5/6< 7/8$. If $n=2$,  then $\varepsilon_{2n-1}=-1$; evaluating $Q_2$ (and using $Q_0=1$) we find $p_{\rm av}(3)\leq 8/27<7/8$. Finally, if $n=1$, we are allowed but a single query; as   ${\rm span}\{\sdket{C}^{\ot t=1}\,:\,C\in\mathsf{Cl}_1\}\subset \mbc^4$, we have from Eq.~\eqref{eq:orbprob} that $p_{\rm av}(1)\leq 4/|\mathsf{Cl}_1|=1/6<7/8$. This concludes the argument.

\section{Further minutiae}\label{sec:minutiae}
\noindent
In this appendix we provide the  proofs of the remaining technical lemmas (which for convenience we restate).
\fullsupport*
\begin{proof} 
The $t-1$ differences in the rank calculation of Eq.~\eqref{eq:affrank} are independent and uniform. The probability that the resulting $n\times (t-1)$ matrix has full row rank is $a_{n,t}$; indeed suppose the first $j$ rows are already linearly independent. Then they span a $j$-dimensional subspace of $\mathbb{F}_2^{t-1}$, which has $2^j$ elements. The next row is uniform on $\mathbb F_2^{t-1}$ and extends the independent set iff it is outside that span, which happens with probability $1-\frac{2^j}{2^{t-1}}$. Multiplying all the conditional probabilities gives:
\[
a_{n,t}=\prod_{j=0}^{n-1}(1-2^{j-(t-1)}).
\]

Now, $1-a_{n,t}$ is the probability that the $t-1$ vectors do not fully span  $\mbf_2^n$, thereby each living in the same proper hyperplane of dimension  $n-1$ (of which there are $(2^n-1)$). A uniformly random vector belongs to such a hyperplane with probability 1/2; the probability that each  vector falls into the same such hyperplane is therefore $(2^{-(t-1)})$. We can therefore union bound to conclude that
\[
1-a_{n,t}\leq (2^n-1)(2^{-(t-1)})\leq 2^{n-t+1}.
\]

\end{proof}

\fullprob*
\begin{proof} 
We recall that there are $K_n=2^n\prod_{j=1}^n (2^j+1)$ total stabilizer states~\cite{aaronson2004improved} and (as a result of their bijection with elements of $\Upsilon$) $2^{(n^2+3n)/2}$ full-support stabilizer states. We can then readily calculate
\begin{equation}
p_n=\frac{2^{(n^2+3n)/2}}{2^n\prod_{j=1}^n (2^j+1)}= \frac{2^{n(n+1)/2}}{\prod_{j=1}^n (2^j+1)}= \frac{\prod_{j=1}^n 2^j}{\prod_{j=1}^n (2^j+1)}= \prod_{j=1}^n\frac{1}{1+2^{-j}}\, ,
\end{equation}
where we have used $\sum_{j=1}^n j = n(n+1)/2$.
Thus,
\[
1/p_n=\prod_{j=1}^n (1+2^{-j})\leq e^{\sum_{j=1}^n2^{-j}}=e^{1-2^{-n}}<e<3.
\]
where we use the fact that for any real number $x$, we have $1+x\leq e^x$.
\end{proof}

\fourier*
\begin{proof}
Recall that from the main text that   the relevant success probability is 
$   \big(\sum_h \sqrt{p^{\text{aff}}_h}\big)^2/|\Upsilon|$, where
 $p_h^{\mathrm{aff}}=|\mathcal F_h|/(a_{n,t}2^{nt})$. So, proving this lemma becomes an exercise in trying to understand some properties of this probability distribution.\\

We begin our proof with some preliminary calculations. First, let us choose $\ket S\in \text{Stab}_n$ and let $L_n(k)$ be the number of $\ket T\in \text{Stab}_n$ such that $\lvert\braket{S|T}\rvert^2=2^{-k}$.  
Notably, this count does not depend on the choice of $\ket S$, as Cliffords act transitively on stabilizer states. Thus, in calculating $L_n(k)$, we can take $\ket S=\ket{0^n}$. A stabilizer state has nonzero overlap with $\ket{0^n}$ when its affine computational-basis support contains $\ket {0^n}$;  its support is then a linear subspace $L\leq \mathbb{F}_2^n$, and its overlap with $\ket{0^n}$ is $2^{-\dim(L)/2}$. So, for a squared overlap of $2^{-k}$, we can take any $k$-dimensional subspace of $\mbf_2^n$, of which there are $\binom nk_2$, where $\binom\cdot\cdot_2$ indicates the Gaussian binomial coefficient. After choosing a basis of $L$, the phase function has the form
    \begin{equation}
        f(x)=\sum_{a=1}^k \alpha_a x_a + 2\sum_{1\leq a<b\leq k} \beta_{ab}x_ax_b\mod 4,
    \end{equation}
    where there are $4^k$ choices of linear phase $\alpha$ and $2^{\binom k2}$ choices of quadratic cross terms $\beta$. Thus in total, the number of phase functions of a fixed support is:
    \begin{equation}
        4^k2^{\binom k2}=2^{2k+k(k-1)/2}=2^{(k^2+3k)/2}\,,
    \end{equation}
so that
\begin{equation}
    L_n(k)=\binom nk_2 2^{(k^2+3k)/2}. \label{eq:lnk}
\end{equation}
Next we define
\begin{equation}
    C_{n,t}:=\sum_{k=1}^n L_n(k)2^{-kt}=\sum_{k=1}^n \left(\binom nk_2 2^{(k^2+3k)/2}\right) 2^{-k(n+s)}\,,
\end{equation}
which is manifestly  the total overlap with a given stabilizer state (after taking $t$ copies) of all the stabilizer states. It will turn out to be useful to bound $C_{n,t}$; to that end, we begin with the Gaussian binomial term. We have:
    \begin{align}
        \binom nk_2 &:= \prod_{j=0}^{k-1}\frac{2^n-2^j}{2^k-2^j}\\
        &= 2^{k(n-k)}\prod_{j=0}^{k-1}\frac{1-2^{j-n}}{1-2^{j-k}}\\
        &\leq 2^{k(n-k)}\prod_{r=1}^k \frac{1}{1-2^{-r}}\\
        &\leq 2^{k(n-k)}\cdot 2\prod_{j=2}^\infty \frac{1}{1-2^{-j}}\\
        &\leq 2^{k(n-k)}\cdot \frac{2}{1-\sum_{j=2}^\infty 2^{-j}}\\
        &\leq 4\cdot 2^{k(n-k)}\,.
    \end{align}
    Thus, we can bound:
    \begin{equation} \label{cns}
        C_{n,n+s} =\sum_{k=1}^n 2^{-k(n+s)}\binom nk_2 2^{(k^2+3k)/2}\leq \sum_{k=1}^n 4\cdot 2^{-k(n+s)}2^{k(n-k)}2^{(k^2+3k)/2}\leq \sum_{k=1}^n 4\cdot 2^{k(1-s)}\leq \frac{4}{2^{s-1}-1}\,,
    \end{equation}
and so we conclude that taking a few more than $n$ copies implies that the total ($t$-fold) overlap of a given stabilizer state with all of the others vanishes very quickly.
Next, we define:
\begin{equation}
    \hat{p}(\alpha,\beta):=\sum_{h\in \widehat\Upsilon} p^{(t)}_h \chi_{h}((\alpha,\beta))=\frac{1}{2^{nt}}\sum_{\x} \chi_{h(\x)}((\alpha,\beta))=\bra{\phi_{0,0}}^{\otimes t}\ket{\phi_{\alpha,\beta}}^{\otimes t}=\braket{\phi_{0,0}|\phi_{\alpha,\beta}}^{ t} \,
\end{equation}
We then also have:
\begin{align}
    \sum_{(\alpha,\beta)}|\hat p(\alpha,\beta)|^2&= \sum_{(\alpha,\beta)}\left(\sum_h p^{(t)}_h \chi_{h}((\alpha,\beta))\right)\left(\sum_{h'}p^{(t)}_{h'}\overline{\chi_{h'}((\alpha,\beta))}\right)\\
    &= \sum_{h, h'} p^{(t)}_h p^{(t)}_{h'} \sum_{(\alpha, \beta)} \chi_{h}((\alpha,\beta)) \overline{\chi_{h'}((\alpha,\beta))}\\
    &= \sum_h (p^{(t)}_h)^2|\Upsilon|,
\end{align}
where the last equation follows from character orthogonality.
Since the $\phi_{\alpha,\beta}$ only correspond to the stabilizer states with full support, we have a bound using the set of all stabilizer states:
\begin{align}
|\Upsilon|\sum_h (p^{(t)}_h)^2&=\sum_{(\alpha,\beta)\in \Upsilon} \lvert\braket{\phi_{0,0}|\phi_{\alpha,\beta}}\rvert^{2t}\\
&\leq \sum_{T\in \text{Stab}_n}\lvert\braket{\phi_{0,0}|T}\rvert^{2t}\\
&=1+C_{n,t},
\end{align}
where we recall from above that $C_{n,t}=\sum_{k=1}^n L_n(k)2^{-kt}$ is the total overlap with $\ket{\phi_{0,0}}$ (after taking $t$ copies) of the stabilizer states not equal to $\ket{\phi_{0,0}}$ itself.  
To include the full rank condition again, let $N_h^{\text{aff}} $ only count the number of affine full-rank tuples. Since an $a_{n,t}$-fraction of the $2^{nt}$ tuples pass the affine-rank test, $\sum_h N_h^{\mathrm{aff}}=a_{n,t}2^{nt}$.
Thus, the conditional probabilities are $p_h^{\mathrm{aff}}= N_h^{\mathrm{aff}}/(a_{n,t}2^{nt})$ and
\begin{equation}
|\Upsilon|\sum_h(p^{\text{aff}}_h)^2\leq|\Upsilon|\sum_h \frac{p_h^2}{a_{n,t}^2} \leq \frac{1+C_{n,t}}{a_{n,t}^2}\,.
\end{equation}
\noindent
Finally, we recall that for any probability distribution $\{r_h\}_h$ one can use Hölder's inequality to deduce
\begin{equation}
    1=\sum_h r_h = \sum_h \bigl(r_h^{1/2}\bigr)^{2/3}\bigl(r_h^{2}\bigr)^{1/3}
\le \Bigl(\sum_h r_h^{1/2}\Bigr)^{2/3}\Bigl(\sum_h r_h^{2}\Bigr)^{1/3}\,,
\end{equation}
\noindent
so that
\begin{equation}
      \Big(\sum_h \sqrt{r_h}\Big)^2\geq \frac{1}{\sum_h r_h^2}\,.
\end{equation}
\noindent
Putting everything together, we have the bound
\begin{equation}
p_{\rm succ} =     \frac{\Big(\sum_h \sqrt{p^{\text{aff}}_h}\Big)^2}{|\Upsilon|}\geq \frac{1}{|\Upsilon|\sum_h (p^{\text{aff}}_h)^2}\geq \frac{a_{n,t}^2}{1+C_{n,t}}\geq \frac{(1-2^{1-s})^2}{1+4/(2^{s-1}-1)}\,,
\end{equation}
where the last inequality follows from Eq.~\eqref{cns} and Lemma~\ref{lem:fullsupp}, which concludes the proof of the lemma.
\end{proof}

\bounds*
\begin{proof}
This is mostly a matter of putting together the results of a bunch of lemmas. 
To begin, recall that during the search for a full-support state, every non-full-support rejection is nondestructive, and therefore consumes no copies of the unknown state. The probability that   none of the $r=\lceil 3\ln(40/\delta)\rceil$ sampled Cliffords gives a full-support chart  is, by  Lemma~\ref{lem:fullprob},  
\begin{equation}
(1-p_n)^r\leq \Big(\frac{2}{3}\Big)^r\leq e^{-r/3}= e^{-\lceil 3 \ln(40/\delta)\rceil/3}\leq\frac{\delta}{40}.
\end{equation}
If this occurs we declare the algorithm to have failed; supposing not, we continue.\\

At the first full-support state, the rank test accepts, by Lemma~\ref{lem:fullsupp}, with probability $a_{n,t}>1-2^{n-t+1}=1-2^{1-s}$, and then the QFT gives the correct $(\alpha,\beta)$, by Lemma~\ref{lem:fourier}, with probability $p_{\rm succ}\geq {(1-2^{1-s})^2}/({1+4/(2^{s-1}-1)})$. 
Thus, we succeed with probability at least
\begin{equation}
P:=(1-(1-p_n)^r)\cdot a_{n,t}\cdot\frac{(1-2^{1-s})^2}{1+4/(2^{s-1}-1)}.
\end{equation}
The failure probability is then bounded using the fact that $1-xyz=(1-x)+x(1-y)+xy(1-z)\leq (1-x)+(1-y)+(1-z)$ for $x,y,z\in [0,1]$:
\begin{align}
    1-P &= 1-(1-(1-p_n)^r)\cdot a_{n,t}\cdot\frac{(1-2^{1-s})^2}{1+4/(2^{s-1}-1)}\\
    &\leq (1-(1-(1-p_n)^r))+(1-a_{n,t})+(1-\frac{(1-2^{1-s})^2}{1+4/(2^{s-1}-1)})\\
    &\leq (1-p_n)^r+2^{n-t+1}+\frac{4/(2^{s-1}-1)+2\cdot2^{1-s}-(2^{1-s})^2}{1+4/(2^{s-1}-1)}\\
    &\leq \frac{\delta}{40} +2^{1-s}+\frac{4}{(2^{s-1}-1)}+2\cdot2^{1-s}\\
    &=\frac{\delta}{40}+3\cdot 2^{1-s}+\frac{4}{(2^{s-1}-1)}
\end{align}
Now, since we chose $t=n+4+\lceil\log_2 1/\delta \rceil$, it follows that $\delta\geq 2^{4-s}$, so that
\begin{equation*}
   \frac{1}{\delta}\big( 3\cdot 2^{1-s}+\frac{4}{2^{s-1}-1}\big) \leq \frac 38 + \frac{4}{8-2^{4-s}}\leq \frac 38 + \frac{4}{7} = \frac{53}{56}\,,
\end{equation*}
from which we conclude $1-P\leq \delta/40+53\delta/56=272\delta/280$.
Finally, choose the isotypic compression accuracy (see  Appendix~\ref{sec:fibers}) to be $\zeta=\delta/40$. Before isotypic compression, the state is $\sum_h c_h \ket{\mcf_h}\ket{h}$ for some $\sum_h|c_h|^2=1$.
For each isotypic, the  compression error is bounded as
\begin{equation}
    \bigl\|U_h^\dagger\ket{\mcf_h}-\ket{0^{nt}}\bigr\|_2 = \bigl\|\ket{\mcf_h}-U_h\ket{0^{nt}}\bigr\|_2\leq \zeta.
\end{equation}
Since the different $h$ branches are orthogonal, the error for the entire isotypic compression step also has bounded 2-norm:
\begin{equation}
    \Bigl\|\sum_h c_h(U_h^\dagger\ket{\mcf_h}-\ket{0^{nt}})\ket{h}\Bigr\|_2^2 = \sum_h |c_h|^2\norm{(U_h^\dagger\ket{\mcf_h}-\ket{0^{nt}})\ket{h}}_2^2 \leq \sum_h |c_h|^2\zeta^2 = \zeta^2.
\end{equation}
Thus, the 2-norm of the difference between the ideal and prepared state is at most $\zeta$, so the trace distance between their density matrices is also at most $\zeta$; indeed, for pure states $\ket{\psi}$ and $\ket{\phi}$, recall that the trace distance is $T(\phi,\psi)=\sqrt{1-\lvert\braket{\phi|\psi}\rvert^2},$ and that the squared 2-norm is $\norm{\ket{\phi}-\ket{\psi}}_2^2 =2-2\Re\braket{\phi|\psi}$; we therefore have 
\begin{equation}
    \norm{\ket{\phi}-\ket{\psi}}_2^2-T(\phi,\psi)^2 = 2-2\Re\braket{\phi|\psi}-(1-|\braket{\phi|\psi}|^2) = |1-\braket{\phi|\psi}|^2\geq 0.
\end{equation}
Now, the rest of the (CPTP) operations of the algorithm cannot increase the trace distance. When measuring, this means the resulting classical outcome distributions for the ideal and implemented cases have total variation distance at most $\zeta$. The total variation distance bounds the difference in probability of any particular event, including the event that the learner fails. Thus, the failure probability of the implemented learner can exceed the ideal learner by at most $\zeta$, and is therefore upper bounded by  ${272\delta}/{280}+\zeta={272\delta}/{280}+{\delta}/{40}={279\delta}/{280}< \delta$. Finally, $\log(1/\zeta)=\mco(\log(1/\delta))$, so that by Lemma~\ref{lem:unitary} the circuit size is polynomial in $n$ and $\log(1/\delta)$.

\end{proof}

\compbound*
\begin{proof}
To begin, a general $n$-qubit Clifford can be synthesised with $\mco(n^2)$ elementary Clifford gates and circuit depth $\mco(n/\log n)$~\cite{aaronson2004improved}; we synthesise $t$  copies in parallel, using $\mco(n^2t)$ gates. Next, the affine-rank test   uses reversible Gaussian elimination on an $n\times (t-1)$ binary matrix, which can be implemented using $\mco(n^2t)$ elementary Boolean operations and $\mco(n^2t)$ circuit depth~\cite{albrecht2011efficient,perriello2021a} and   $\mathcal{O}(n^2)$   ancilla qubits~\cite{perriello2021a}. In total then, and as by Lemma~\ref{lem:bounds} it suffices to take $r=\mco(\log(1/\delta))$ sequential trials, the search for the full-support coordinate system uses $\mco(n^2t\log(1/\delta))$ gates, and likewise takes time $\mco(n^2t\log(1/\delta))$.
The isotypic compression step is more costly. By Lemma~\ref{lem:unitary},  our isotypic compression unitary can be implemented to error $\zeta$ in depth $\mco(nt(t+\log(nt/\zeta))^3)$; as we saw in the proof of Lemma~\ref{lem:bounds}, it suffices to take $\zeta\sim\delta$, so that this stage requires depth $\mco(nt(t+\log(nt/\delta))^3)$ (and $(t+\log(nt/\delta))^3$ ancillae to make various computations reversible). 
Now, because $\Upsilon\cong\mbz_4^n\times\mathbb{F}_2^{n(n-1)/2}$, its Fourier transform factorizes as $\mathsf{F}_\Upsilon=
(\mathsf{F}_{\mbz_4})^{\otimes n}\otimes H^{\otimes\binom n2}$~\cite{hoyer1997efficient}, using therefore  $\mco(n^2)$ gates in $\mco(1)$ depth~\cite{cleve2000fast}. Finally, the recovery of the stabilizer tableau of the unknown state can be achieved, as explained in the main text, in time $\mco(n^3)$~\cite{aaronson2004improved}.
So the complexity is dominated by the orbit compression step. Finally, and again by Lemma~\ref{lem:bounds}, we can take $t=\mco(n+ \log(1/\delta))$, so that $t+\log(nt/\delta)\in \mco(t)$, and our total complexity is $\mco(nt^4)=\mco(n(n+\log(1/\delta))^4)$.
\end{proof}

\end{document}